\documentclass[lettersize,journal]{IEEEtran}
\usepackage[latin9]{inputenc}
\usepackage{float}
\usepackage{bm}
\usepackage{amsmath}
\usepackage{amsthm}
\usepackage{graphicx}

\makeatletter
\theoremstyle{plain}
	  \newtheorem{prop}{Proposition}
\ifx\proof\undefined\
  \newenvironment{proof}[1][\proofname]{\par
    \normalfont\topsep6\p@\@plus6\p@\relax
    \trivlist
    \itemindent\parindent
    \item[\hskip\labelsep
          \scshape
      #1]\ignorespaces
  }{%
    \endtrivlist\@endpefalse
  }
  \providecommand{\proofname}{Proof}
\fi
\theoremstyle{remark}
\newtheorem{rem}{Remark}

\usepackage{amsfonts}\usepackage{algorithmic}
\usepackage{cuted}
\usepackage{multicol}
\usepackage{algorithm}
\usepackage{array}
\usepackage{float}
\usepackage{url}
\usepackage{verbatim}
\usepackage{cite}
\usepackage[acronym]{glossaries}
\usepackage{amsmath, amssymb}
\newcommand{\newac}{\newacronym}
\newcommand{\ac}{\gls}

\makeglossaries
\newac{isac}{ISAC}{integrated sensing and communication}
\newac{ula}{ULA}{uniform linear array}
\newac{los}{LoS}{line-of-sight}
\newac{rcs}{RCS}{radar cross section}
\newac{aod}{AoD}{angle of departure}
\newac{aoa}{AoA}{angle of arrival}
\newac{dofs}{DoFs}{degrees of freedom}
\newac{cu}{CU}{communication user}
\newac{sdr}{SDR}{semi-definite relaxation}
\newac{svd}{SVD}{singular value decomposition}
\newac{crb}{CRB}{Cram{\'e}r-Rao bound}
\newac{sinr}{SINR}{signal-to-interference-plus-noise ratio}
\newac{snr}{SNR}{signal-to-noise ratio}
\newac{mimo}{MIMO}{multiple-input-multiple-output}
\newac{cscg}{CSCG}{circularly symmetric complex Gaussian}
\makeatother

\begin{document}
\title{Secure Cell-Free Integrated Sensing and Communication in the Presence
of Information and Sensing Eavesdroppers}
\author{Zixiang Ren, \textit{Gradutae Student Member, IEEE}, Jie Xu, \textit{Senior
Member, IEEE,} Ling Qiu, \textit{Member, IEEE,} and Derrick Wing Kwan
Ng, \textit{Fellow, IEEE}\thanks{Z. Ren is with the Key Laboratory of Wireless-Optical Communications,
Chinese Academy of Sciences, School of Information Science and Technology,
University of Science and Technology of China, Hefei 230027, China,
and the Future Network of Intelligence Institute (FNii), The Chinese
University of Hong Kong (Shenzhen), Shenzhen 518172, China (e-mail:
rzx66@mail.ustc.edu.cn).}\thanks{L. Qiu is with the Key Laboratory of Wireless-Optical Communications,
Chinese Academy of Sciences, School of Information Science and Technology,
University of Science and Technology of China, Hefei 230027, China
(e-mail: lqiu@ustc.edu.cn).}\thanks{J. Xu is with the School of Science and Engineering (SSE) and the
FNii, The Chinese University of Hong Kong (Shenzhen), Shenzhen 518172,
China (e-mail: xujie@cuhk.edu.cn).}\thanks{D. W. K. Ng is with the University of New South Wales, Sydney, NSW
2052, Australia (e-mail: w.k.ng@unsw.edu.au).}\thanks{L. Qiu and J. Xu are the corresponding authors.}\vspace{-0.3cm}}
\maketitle
\begin{abstract}
This paper studies a secure cell-free integrated sensing and communication
(ISAC) system, in which multiple ISAC transmitters collaboratively
send confidential information to multiple communication users (CUs)
and concurrently conduct target detection. Different from prior works
investigating communication security against potential information
eavesdropping, we consider the security of both communication and
sensing in the presence of information and sensing eavesdroppers
that aim to intercept confidential communication information and extract
target information, respectively. Towards this end, we optimize the
joint information and sensing transmit beamforming at these ISAC transmitters
for secure cell-free ISAC. Our objective is to maximize the detection
probability over a designated sensing area while ensuring the minimum
signal-to-interference-plus-noise-ratio (SINR) requirements at CUs.
Our formulation also takes into account the maximum tolerable signal-to-noise
ratio (SNR) constraints at information eavesdroppers for ensuring the confidentiality
of information transmission, and the maximum detection probability
constraints at sensing eavesdroppers for preserving sensing privacy.
The formulated secure joint transmit beamforming problem is highly
non-convex due to the intricate interplay between the detection probabilities,
beamforming vectors, and SINR constraints. Fortunately, through strategic
manipulation and via applying the semidefinite relaxation (SDR) technique,
we successfully obtain the globally optimal solution to the design
problem by rigorously verifying the tightness of SDR. Furthermore,
we present two alternative joint beamforming designs based on the
sensing SNR maximization over the specific sensing area and the coordinated
beamforming, respectively. Numerical results reveal the benefits of
our proposed design over these alternative benchmarks.
\end{abstract}


\section{Introduction}

Integrated sensing and communication (ISAC) has been identified as
one of the six delineated usage scenarios for future sixth-generation
(6G) wireless networks \cite{itu2030}, which holds the capability
to support a variety of new applications, such as navigation, activity
recognition, environment monitoring, and sensing data acquisition
\cite{khan2022digital,li2023csi}. As a result, ISAC has recently
emerged as one of the hottest topics within the wireless communication
community, spurring extensive research and development \cite{liu2021integrated,liu2023seventy,masouros2021introduction}.
The exploration of ISAC for enhancing both sensing and communication
performances spans different technical perspectives, including fundamental
information theoretic limits \cite{xiong2023fundamental,Haocheng2022,ouyang2023integrated},
transmit waveform design \cite{xiao2022waveform}, beamforming optimization
\cite{liu2020joint,9652071}, active sensing \cite{sohrabi2022active},
and network architectures \cite{RahLusJ20}. In recent years, the
advancements in the multi-antenna technology have significantly enhanced
ISAC performance. In particular, the deployment of multiple antennas
at ISAC transmitters not only provides multiplexing and diversity
gains for substantially enhancing the communication rate and reliability,
but also offers additional degrees of freedom (DoFs) for refining
sensing resolution and accuracy. Furthermore, besides reusing information
beams for the dual sensing purpose, dedicated sensing beams can be
additionally exploited to provide full available sensing DoFs. As
such, the joint information and sensing beamforming design has emerged
as a promising ISAC solution \cite{LiuMasJ18,9652071,Bruno,wang2023globally}.

While initial ISAC research focused on the single-cell scenarios with
a single ISAC transmitter, future 6G wireless networks are expected
to incorporate densely deployed base stations (BSs). With the advancements
in coordinated multi-point (CoMP) transmission \cite{gesbert2010multi},
cloud radio access network \cite{wu2015cloud}, and cell-free multiple-input
multiple-output (MIMO) \cite{bjornson2020scalable}, leveraging multiple
BSs as cooperative ISAC transmitters serves as a promising natural
architecture in further enhancing performance. On the one hand, there
have been several works investigating coordinated beamforming among
multiple BSs for enabling networked ISAC \cite{xu2023integrated,cheng2023optimal,huang2022coordinated},
in which different ISAC transmitters send independent information
and sensing signals to communicate with their respective communication
users (CUs) and perform joint target detection, estimation, or localization
via multi-static or distributed MIMO sensing \cite{liu2022networked}.
By cooperatively designing the coordinated beamforming vectors, ISAC
transmitters not only effectively mitigate interference among
CUs but also achieve enhanced cooperative multi-static sensing. For
instance, \cite{xu2023integrated} introduced a novel approach for
coordinated ISAC in cellular networks, considering a beampattern optimization
problem subject to communication signal-to-interference-plus-noise-ratio
(SINR) constraints and sensing receive power constraints. Furthermore,
the authors in \cite{cheng2023optimal} explored a multi-antenna networked
ISAC system, maximizing the detection probability under communication
SINR requirements and power constraints via jointly optimizing the
information and sensing beamforming. In \cite{huang2022coordinated},
the authors considered the total power minimization problem in a networked
ISAC system by collaboratively designing power control for different
BSs. 

On the other hand, the utilization of cell-free MIMO in ISAC has emerged
as another viable realization of networked ISAC \cite{cellfree22survey},
where BSs or ISAC transmitters are connected to a central controller
to share the communication and sensing data for joint transmission
and collaborative sensing information processing \cite{cellfree22survey}.
Different from coordinated beamforming, cell-free ISAC can achieve
enhanced communication performance by transforming harmful inter-cell
interference into a part of useful information signals, and improve
sensing performance via advanced sensing signal processing \cite{bjornson2020scalable}.
Inspired by these advantages, recent studies have investigated cell-free
ISAC systems from different perspectives \cite{Joint23Yu,demirhan2023cell,behdad2023multi,zeng2023integrated}.
For example, the authors in \cite{Joint23Yu} maximized the sum of
communication and sensing rates by optimizing user association and
power allocation, adopting a conjugate beamforming approach. Furthermore,
in \cite{behdad2023multi}, the authors explored the multi-static
sensing for cooperative target detection in cell-free ISAC, in which
power allocation at different BSs is jointly optimized to maximize
the sensing signal-to-noise ratio (SNR) while ensuring certain communication
SINR requirements. Moreover, the authors in \cite{demirhan2023cell}
optimized transmit information and sensing beamforming vectors at
different BSs, in which the sensing SNR is maximized subject to the
communication SINR constraints and the individual transmit power limitations
at BSs.

The emergence of ISAC networks, however, introduces severe security
concerns in both communication and sensing. First, to facilitate sensing,
the optimized transmit information beams in ISAC systems are deliberately
aimed at sensing targets to enhance sensing performance, thereby posing
a potential risk of information leakage. This risk is particularly
critical when sensing targets include suspicious entities such as
eavesdropping unmanned aerial vehicles (UAVs) or other adversarial
agents, escalating the potential of unauthorized interception of transmitted
information. To address these concerns, several recent works have
been developed to safeguard against unintended information leakage
\cite{wei2022toward,SuLiuChrJ21,xu2022robust,ren23robust}. For example,
to tackle communication security issues in ISAC systems, the authors
in \cite{wei2022toward} explored the interplay between ISAC and secure
communications to enable a multi-function wireless network integrating
sensing, communication, and security. Moreover, the authors in \cite{SuLiuChrJ21}
studied a secure ISAC system with a single eavesdropping target and
multiple CUs, by considering line-of-sight (LoS) channel models with
angle uncertainty for the eavesdropping target. The objective is to
minimize the eavesdropping SINR at the target while satisfying the
requirements for communication SINR at CUs and sensing beampatterns.
Additionally, \cite{xu2022robust} proposed an optimization framework
for robust secure resource allocation in a secure ISAC system, jointly
optimizing transmit beamforming and snapshot length, accounting for
target angle uncertainty. In a related vein, the authors in \cite{ren23robust}
further considered the robust secure transmit beamforming problem
for a single ISAC transmitter communicating with a single CU and detecting
multiple targets, in which the transmit beampattern distortion is
minimized under secrecy rate constraints for CUs with two different
imperfect CSI scenarios. 

Furthermore, the ISAC systems also encounter new sensing security
threats, as the sensing information might be vulnerable to sensing
eavesdroppers (see, e.g., \cite{dimas2019radar}). By leveraging the
sensing signals of ISAC systems, sensing eavesdroppers in ISAC systems
may silently intercept sensing results without actively transmitting
their own signals. Based on the intercepted sensing information, the
adversary may infer the action of associated physical systems and
possibly launch further actions jeopardizing system performance. Indeed,
this passive eavesdropping on sensing information introduces additional
privacy and security challenges, necessitating advanced mechanisms
to ensure confidentiality and integrity of the sensed data \cite{da2023multi,da2023privacy}.
For instance, \cite{da2023multi} investigated the precoder design
in a single ISAC transmitter scenario based on sensing beampattern
distortion. This study introduced a sensing adversary estimation framework
tailored for estimating target location capitalizing on Bayesian
inference. Besides, the authors in \cite{da2023privacy} further considered
the precoder design in a cell-free ISAC system based on sensing SNR
maximization. The study extended the sensing eavesdropper model via
exploiting an expectation maximization method to eavesdrop target
information. However, the prior research has not addressed the aspect
of transmit design to safeguard sensing privacy \cite{da2023privacy,da2023multi}.
Furthermore, there is no existing work considering both communication
and sensing security in cell-free ISAC systems, thus motivating our
work.

This paper investigates a secure cell-free ISAC system, which comprises
multiple ISAC transmitters collaboratively transmitting confidential
information to multiple CUs, while concurrently performing joint target
detection. We consider that there exist both information eavesdroppers
and sensing eavesdroppers in this system, which aim to intercept confidential
communication information and seek to extract sensing target information,
respectively. The main results of this paper are listed as follows.
\begin{itemize}
\item Firstly, we introduce the system model for secure cell-free ISAC systems,
including a communication framework and a multi-static sensing model
that take into account the existence of sensing and information eavesdroppers.
Our setup assumes that sensing receivers are equipped with knowledge
of the transmitted signal, enabling effective clutter signal mitigation.
By contrast, sensing eavesdroppers lack knowledge of the transmit
signals, preventing them from mitigating the impacts caused by sensing
clutters. In this scenario, we derive the detection probability at
sensing receivers by exploiting the generalized likelihood ratio test
(GLRT) detector. Besides, by assuming that eavesdroppers exploit
signal power for target detection due to the lack of signal knowledge,
we derive the closed-form eavesdropping probability underlining the
interplay with different parameters.
\item Next, we formulate the detection probability maximization problem,
subject to the minimum SINR constraints at CUs for ensuring the successful
transmission of confidential information. Meanwhile, the maximum tolerable
SNR constraints at information eavesdroppers and the maximum eavesdropping
probability constraints at sensing eavesdroppers are considered to
safeguard information and sensing privacy, respectively. The formulated
design problem, however, is highly difficult to solve, due to the
inherent intractability caused by complex relationships between transmit
beamforming vectors and legal sensing receivers/illegal sensing eavesdroppers,
as well as the non-convex nature of communication SINR constraints.
Fortunately, we achieve a globally optimal solution through a meticulously
devised three-step approach. Initially, we reformulate the detection
probabilities for legal sensing receivers and the eavesdropping probabilities for
illegal sensing eavesdroppers to facilitate problem tractability.
Subsequently, we relax the beamforming design problem by exploiting a
semidefinite relaxation (SDR) approach \cite{luo2010semidefinite},
leading to a convex version that can be optimally solved with off-the-shelf
toolboxes. Finally, rigorous proof of the relaxation's tightness is
presented to verify the global optimality of the obtained solution.
\item On the other hand, to cope with the needs in different practical scenarios,
we present two alternative transmit beamforming designs based on the
sensing SNR maximization and the coordinated beamforming, respectively.
For the sensing SNR maximization design, our goal is to maximize the
sensing power at the target direction, while in the coordinated beamforming
design, each CU is associated with a specific BS for independent signal
transmission. %Optimal solutions are proposed for both alternative
designs.
\item Finally, numerical results are provided to validate the effectiveness
of our proposed design, with comparisons against benchmarking sensing
SNR maximization and coordinated beamforming. It is shown that via
joint signal processing in the central controller, the proposed cell-free
design effectively exploits the signal correlation among different
transmitters, and also strategically
utilizes the inherent sensing clutters to jam sensing eavesdroppers,
thus ensuring the sensing security while improving the detection probability.
\end{itemize}

The remainder of this paper is organized as follows. Section II introduces
the system model. Section III derives the detection probability and
eavesdropping probability at sensing receivers and sensing eavesdroppers,
respectively. Section IV formulates the joint transmit beamforming problem for secure cell-free ISAC, and develops
a globally optimal solution to the formulated problem. Section
V presents two alternative design approaches based on SNR maximization
and coordinated beamforming, respectively. Section VI presents numerical
results. Finally, Section VII concludes this paper.

\textit{Notations}: Vectors and matrices are denoted by bold lower-
and upper-case letters, respectively. $\mathbb{C}^{N\times M}$ denotes
the space of $N\times M$ complex matrices. $\boldsymbol{I}$ and
$\boldsymbol{0}$ represents an identity matrix and an all-zero matrix
with appropriate dimensions, respectively. For a square matrix $\boldsymbol{A}$,
$\textrm{tr}(\boldsymbol{A})$ denotes its trace and $\boldsymbol{A}\succeq\boldsymbol{0}$
means that $\boldsymbol{A}$ is positive semi-definite. For a complex
arbitrary-size matrix $\boldsymbol{B}$, $\boldsymbol{B}[i,j]$, $\textrm{rank}(\boldsymbol{B})$,
$\boldsymbol{B}^{T}$, $\boldsymbol{B}^{H}$, and $\boldsymbol{B}^{c}$
denote its $(i,j)$-th element, rank, transpose, conjugate transpose,
and complex conjugate, respectively. For a vector $\boldsymbol{a}$,
$\boldsymbol{a}[i]$ denotes its $i$-th element. $\mathbb{E}(\cdot)$
denotes the statistical expectation. $\|\cdot\|$ denotes the Euclidean
norm of a vector. $|\cdot|$, $\mathrm{Re}(\cdot)$, and $\mathrm{Im}(\cdot)$
denote the absolute value, the real component, and the imaginary component
of a complex entry. $\mathcal{CN}(\boldsymbol{x},\boldsymbol{Y})$
denotes a \ac{cscg} random vector with mean vector $\boldsymbol{x}$
and covariance matrix $\boldsymbol{Y}$. $\boldsymbol{A}\otimes\boldsymbol{B}$
represents the Kronecker product of two matrices $\boldsymbol{A}$
and $\boldsymbol{B}$. $\mathrm{blkdiag}(\cdot)$ constructs a block
diagonal matrix with its entities. $\frac{\partial(\cdot)}{\partial}$
denotes the operator of a partial derivative.

\section{System Model }

We consider a secure cell-free ISAC system as shown in Fig. 1, which
comprises $M_{t}$ ISAC transmitters, $M_{r}$ sensing receivers,
$K$ single-antenna CUs, as well as $L$ single-antenna information
eavesdroppers and $Q$ sensing eavesdroppers. Let $\mathcal{M}_{t}\overset{\triangle}{=}\{1,\dots,M_{t}\}$,
$\mathcal{M}_{r}\overset{\triangle}{=}\{1,\dots,M_{r}\}$, $\mathcal{K}\overset{\triangle}{=}\{1,\dots,K\}$,
$\mathcal{L}\overset{\triangle}{=}\{1,\dots,L\}$, and $\mathcal{Q}\overset{\triangle}{=}\{1,\dots,Q\}$
denote the sets of ISAC transmitters, sensing receivers, CUs, information
eavesdroppers, and sensing eavesdroppers, respectively. Without loss
of generality, we assume that each ISAC transmitter, sensing receiver,
and sensing eavesdropper in our system is equipped with an array of
$N$ antennas. 

In this ISAC system, the objective is to address
the communication requirements of the $K$ CUs while simultaneously
conducting sensing operations in a specific area of interest. The
central controller coordinates the ISAC transmitters and sensing receivers
to ensure the security in the cell-free ISAC system. It is assumed
that all ISAC transmitters and sensing receivers achieve perfect synchronization
facilitated by the central controller \cite{behdad2023multi,demirhan2023cell}.
Additionally, we consider a basic scenario in which there is no collaboration
among any sensing or information eavesdroppers.
\begin{figure}[H]
\centering\includegraphics[scale=0.25]{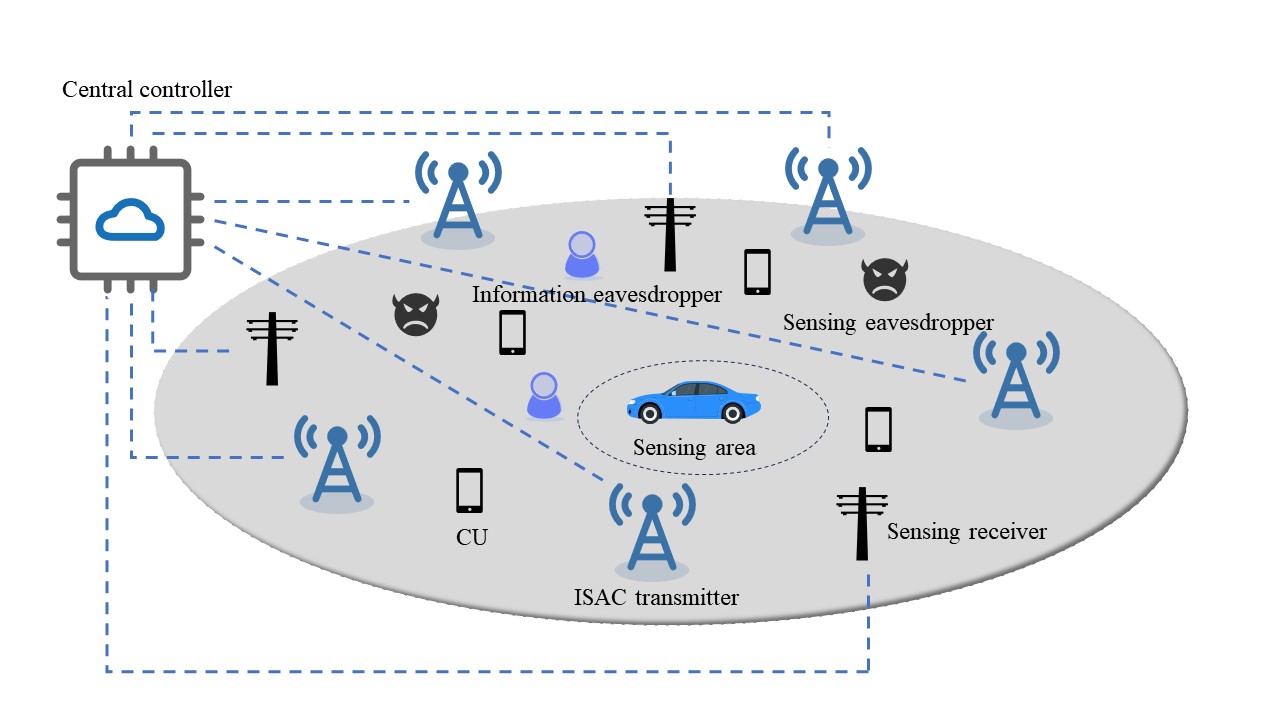}\caption{Illustration of the secure cell-free ISAC system.}
\end{figure}

\subsection{Communication Model}

In our framework, we focus on the ISAC transmission over a period of $T$ symbols, where $T$ is assumed
to be sufficiently large. We denote $\hat{s}_{k}^{\mathrm{I}}(t)\in\mathbb{C}$
as the desired information signal for CU $k\in\mathcal{K}$ in the
$t$-th symbol. We model $\hat{s}_{k}^{\mathrm{I}}(t)$'s
as independent and identically distributed (i.i.d.) CSCG random variables
each with zero mean and unit variance, i.e., $\hat{s}_{k}^{\mathrm{I}}(t)\sim\mathcal{CN}(0,1)$.
Let $\hat{\boldsymbol{w}}_{i,k}\in\mathbb{C}^{N\times1}$ denote the
transmit beamforming vector at transmitter $i\in\mathcal{M}_{t}$
for CU $k$. We define $\boldsymbol{w}_{k}\in\mathbb{C}^{NM_{t}\times1}$
as the beamforming vector spanning all the $M_{t}$ ISAC transmitters
for CU $k\in\mathcal{K}$, i.e., 
\begin{equation}
\boldsymbol{w}_{k}=[\hat{\boldsymbol{w}}_{1,k}^{T},\dots,\hat{\boldsymbol{w}}_{M_{t},k}^{T}]^{T}.
\end{equation}
Moreover, we assume that all ISAC transmitters collaboratively employ
dedicated sensing signals to fully exploit the available DoFs for
the purpose of sensing \cite{liu2020joint}. Let $\hat{\boldsymbol{s}}_{i}^{\mathrm{S}}(t)$
denote the dedicated sensing signal at ISAC transmitter $i\in\mathcal{M}_{t}$.
In this context, we define $\boldsymbol{s}^{\mathrm{S}}(t)\in\mathbb{C}^{NM_{t}\times1}$
as the dedicated sensing signal in the $t$-th symbol, spanning all the $M_{t}$ ISAC transmitters, i.e.,
\begin{equation}
\boldsymbol{s}^{\mathrm{S}}(t)=[\big(\hat{\boldsymbol{s}}_{1}^{\mathrm{S}}(t)\big)^{T},\dots,\big(\hat{\boldsymbol{s}}_{M_{t}}^{\mathrm{S}}(t)\big)^{T}]^{T}.
\end{equation}
 We define the covariance of the dedicated sensing signal $\boldsymbol{s}^{S}(t)$
as 
\begin{equation}
\boldsymbol{S}=\mathbb{E}\big(\boldsymbol{s}^{\mathrm{S}}(t)(\boldsymbol{s}^{\mathrm{S}}(t))^{H}\big).
\end{equation}
Without loss of generality, we assume that $\boldsymbol{S}$ is a
general-rank matrix, serving as an optimization variable in our system.
Typically, the number of dedicated sensing beams is determined by
the rank of $\boldsymbol{S}$. Consequently, the transmitted signal
at ISAC transmitter $i\in\mathcal{M}_{t}$ is expressed as \cite{behdad2023multi}
\begin{eqnarray}
\boldsymbol{x}_{i}(t) & = & \sum_{k=1}^{K}\hat{\boldsymbol{w}}_{i,k}\hat{s}_{k}^{\mathrm{I}}(t)+\hat{\boldsymbol{s}}_{i}^{\mathrm{S}}(t).
\end{eqnarray}
Let $\boldsymbol{x}(t)=[\boldsymbol{x}_{1}^{T}(t),\dots,\boldsymbol{x}_{M_{t}}^{T}(t)]^{T}$
denote the accumulated transmitted signal across all the $M_{t}$ ISAC transmitters.
We define the transmit covariance across all the ISAC transmitters as 
\begin{eqnarray}
\boldsymbol{R} & = & \mathbb{E}\big(\boldsymbol{x}(t)\boldsymbol{x}^{H}(t)\big)=\sum_{k=1}^{K}\boldsymbol{w}_{k}\boldsymbol{w}_{k}^{H}+\boldsymbol{S}.\label{eq:cov}
\end{eqnarray}

Let $\hat{\boldsymbol{h}}_{i,k}\in\mathbb{C}^{N\times1}$ denote the
channel vector between CU $k\in\mathcal{K}$ and ISAC transmitter
$i\in\mathcal{M}_{t}$. Here, we introduce the overall channel from
CU $k\in\mathcal{K}$ to all $M_{t}$ transmitters as $\boldsymbol{h}_{k}\in\mathbb{C}^{NM_{t}\times1}$,
i.e., 

\begin{equation}
\boldsymbol{h}_{k}=[\hat{\boldsymbol{h}}_{1,k}^{T},\dots,\hat{\boldsymbol{h}}_{M_{t},k}^{T}]^{T}.
\end{equation}
 As a result, the received signal at CU $k\in\mathcal{K}$ is expressed
as \eqref{eq:signal} at the top of the next page, which consists of three
main components, i.e., desired signal, multi-user interference, and sensing
signal interference. \begin{figure*}

\centering
\begin{align}
y_{k}(t)=\sum_{i=1}^{M_{t}}\hat{\boldsymbol{h}}_{i,k}^{H}\boldsymbol{x}_{i}(t)+n_{k}(t) & =\underset{\textrm{Desired signal}}{\underbrace{\boldsymbol{h}_{k}^{H}\boldsymbol{w}_{k}\hat{s}_{k}^{I}(t)}}+\underset{\textrm{Multi-user interference}}{\underbrace{\sum_{j=1,j\neq k}^{K}\boldsymbol{h}_{k}^{H}\boldsymbol{w}_{j}\hat{s}_{j}^{I}(t)}}+\underset{\textrm{Sensing signal interference}}{\underbrace{\boldsymbol{h}_{k}^{H}\boldsymbol{s}^{\mathrm{S}}(t)}}+n_{k}(t)\label{eq:signal}
\end{align}
\hrule\end{figure*} Here, $n_{k}(t)$ is the independent Gaussian
noise with a zero mean and variance $\sigma^{2}$ at CU $k\in\mathcal{K}$
in the $t$-th symbol, i.e., $n_{k}(t)\sim\mathcal{CN}(0,\sigma^{2})$.
As such, the received SINR at the receiver of CU $k\in\mathcal{K}$
is given in \eqref{eq:sinr} at the top of the next page. \begin{figure*}
\begin{equation}
\gamma_{k}(\{\boldsymbol{w}_{k}\},\boldsymbol{S})=\frac{|\sum_{i=1}^{M_{t}}\hat{\boldsymbol{h}}_{i,k}^{H}\hat{\boldsymbol{w}}_{i,k}|^{2}}{\sum_{j=1,j\neq k}^{K}|\sum_{i=1}^{M_{t}}\hat{\boldsymbol{h}}_{i,k}^{H}\hat{\boldsymbol{w}}_{i,j}|^{2}+\mathbb{E}\big(|\sum_{i=1}^{M_{t}}\hat{\boldsymbol{h}}_{i,k}^{H}\hat{\boldsymbol{s}}_{i}^{\mathrm{S}}(t)|^{2}\big)+\sigma^{2}}=\frac{|\boldsymbol{h}_{k}^{H}\boldsymbol{w}_{k}|^{2}}{\sum_{j=1,j\neq k}^{K}|\boldsymbol{h}_{k}^{H}\boldsymbol{w}_{j}|^{2}+\boldsymbol{h}_{k}^{H}\boldsymbol{S}\boldsymbol{h}_{k}+\sigma^{2}}.\label{eq:sinr}
\end{equation}

\centering\hrule\end{figure*}It is worth noting that the interference
term in \eqref{eq:sinr} stems from two aspects, i.e., multi-user
interference and sensing signal interference \cite{demirhan2023cell}. 

\subsection{Information Eavesdropping}

In this subsection, we focus on the information eavesdropping model, in which information eavesdropper $l\in\mathcal{L}$
may attempt to intercept confidential information intended for any CU
$k\in\mathcal{K}$. Let $\hat{\boldsymbol{g}}_{i,l}\in\mathbb{C}^{N\times1}$
represent the channel vector from ISAC transmitter $i\in\mathcal{M}_{t}$
to information eavesdropper $l\in\mathcal{L}$. Let $\boldsymbol{g}_{l}\in\mathbb{C}^{NM_{t}\times1}$
denote the accumulated channel from all ISAC $M_{t}$ transmitters to eavesdropper $l$,
i.e., 
\begin{equation}
\boldsymbol{g}_{l}=[\hat{\boldsymbol{g}}_{1,l}^{T},\dots,\hat{\boldsymbol{g}}_{M_{t},l}^{T}]^{T}.
\end{equation}
The received signal at information eavesdropper $l\in\mathcal{L}$
is given by \eqref{eq:eavesignal} at the top of next page, \begin{figure*}

\centering 
\begin{equation}
\tilde{y}_{l}(t)=\sum_{i=1}^{M_{t}}\hat{\boldsymbol{g}}_{i,l}^{H}\boldsymbol{x}_{i}(t)=\boldsymbol{g}_{l}^{H}\boldsymbol{w}_{k}\hat{s}_{k}^{I}(t)+\sum_{j=1,j\neq k}^{K}\boldsymbol{g}_{l}^{H}\boldsymbol{w}_{j}\hat{s}_{j}^{I}(t)+\boldsymbol{g}_{l}^{H}\boldsymbol{s}^{\mathrm{S}}(t)+\tilde{n}_{l}(t).\label{eq:eavesignal}
\end{equation}
\hrule\end{figure*}in which $\tilde{n}_{l}(t)$ denotes the noise
at the receiver that is an i.i.d. CSCG random variable with a zero
mean and variance $\sigma^{2}$. It is assumed that there is no cooperation
among different information eavesdroppers. 

Notice that if information eavesdropper $l\in\mathcal{L}$ is aware
of the channel $\boldsymbol{g}_{l}$ and transmitted sensing signal
sequence $\{\hat{\boldsymbol{s}}_{i}^{\mathrm{S}}(t)\}$, then it
may be able to effectively cancel the interference caused by sensing
signals via advanced signal processing \cite{10086626}. Subsequently, by employing successive interference
cancellation (SIC) \cite{chen2018design}, each information eavesdropper
can proceed to cancel the information signal intended for other CUs (if they are decoded)
before attempting to decode the signal of CU $k\in\mathcal{K}$. As
a result, we impose the worst-case assumption that information eavesdropper
$l\in\mathcal{L}$ can perfectly cancel the information signal for
other CUs and the dedicated sensing signal. In this case and under the
assumption without cooperation among different information eavesdroppers,
the received eavesdropping SNR at information eavesdropper $l\in\mathcal{L}$
for intercepting signals for CU $k\in\mathcal{K}$ is given by
\begin{equation}
\tilde{\gamma}_{l,k}(\{\boldsymbol{w}_{k}\},\boldsymbol{S})=\frac{|\boldsymbol{g}_{l}^{H}\boldsymbol{w}_{k}|^{2}}{\sigma^{2}}.
\end{equation}
To protect the communication security, in this paper we need to ensure
that the information eavesdropping SNR $\tilde{\gamma}_{l,k}(\{\boldsymbol{w}_{k}\},\boldsymbol{S})$ at any information eavesdropper $l\in\mathcal{L}$
should not exceed a given threshold for all CU $k\in\mathcal{K}$.

\subsection{Multi-static Sensing}

In this subsection, we consider the multi-static sensing model within
this cell-free network. In this scenario, the central controller aggregates
the received signals from all $M_{r}$ receivers to perform joint
target detection. First, we assume that the sensing receivers are
aware of the transmitted signal $\boldsymbol{x}_{i}(t)$ and the environmental
information, including clutters information introduced by stationary
objects and LoS path. As such, the sensing receivers possess
the capability to efficiently mitigate signals originating from clutters
and LoS path \cite{behdad2023multi}. In case that the target is present,
the received signal at sensing receiver $j\in\mathcal{M}_{r}$ in
the $t$-th symbol is expressed as 
\begin{equation}
\boldsymbol{r}_{j}(t)=\sum_{i=1}^{M_{t}}\alpha_{i,j}\boldsymbol{a}_{r}(\varphi_{j})\boldsymbol{a}_{t}^{H}(\theta_{i})\boldsymbol{x}_{i}(t)+\bar{\boldsymbol{n}}_{j}(t),
\end{equation}
where $\alpha_{i,j}\in\mathbb{C}$ represents the complex coefficient
characterizing the influence of path-loss and the target \ac{rcs} between ISAC
transmitter $i\in\mathcal{M}_{t}$ and sensing receiver $j\in\mathcal{M}_{r}$,
which is an unknown deterministic parameter. Additionally, $\theta_{i}$
denotes the \ac{aod} from ISAC transmitter $i\in\mathcal{M}_{t}$
to the target and $\varphi_{j}$ denotes the \ac{aoa} from the target
to sensing receiver $j\in\mathcal{M}_{r}$. Furthermore, $\boldsymbol{a}_{t}(\cdot)$
and $\boldsymbol{a}_{r}(\cdot)$ represent the transmit and receive
steering vectors, respectively, and $\bar{\boldsymbol{n}}_{j}$ represents
the Gaussian noise at sensing receiver $j\in\mathcal{M}_{r}$, with
each element having a zero mean and a variance of $\sigma_{s}^{2}$.
For notational convenience, we denote 
\begin{equation}
\boldsymbol{\phi}_{i,j}(t)\overset{\triangle}{=}\boldsymbol{a}_{r}(\varphi_{j})\boldsymbol{a}_{t}^{H}(\theta_{i})\boldsymbol{x}_{i}(t)
\end{equation}
as the signal received by sensing receiver $j\in\mathcal{M}_{r}$
from ISAC transmitter $i\in\mathcal{M}_{t}$ reflected by the target
excluding the influence of $\alpha_{i,j}$. It is worth noting that
$\boldsymbol{\phi}_{i,j}(t)$ is assumed to be perfectly known at
sensing receivers and thus can be utilized to design a target detector.
Let $\boldsymbol{\Phi}_{j}(t)=[\boldsymbol{\phi}_{1,j}(t),\dots,\boldsymbol{\phi}_{M_{t},j}(t)]\in\mathbb{C}^{N\times M_{t}}$
denote the received signal at sensing receiver $j\in\mathcal{M}_{r}$
from all ISAC transmitters. The reflected signal at all sensing receivers
from all ISAC transmitters in the $t$-th symbol is given as
\begin{equation}
\boldsymbol{\Phi}(t)=\mathrm{blkdiag}(\boldsymbol{\Phi}_{1}(t),\dots,\boldsymbol{\Phi}_{M_{r}}(t))\in\mathbb{C}^{NM_{r}\times M_{t}M_{r}}.\label{eq:Y}
\end{equation}
As a result, the concatenated received signal over all $M_{r}$ sensing
receivers at the $t$-th symbol is given by
\begin{equation}
\boldsymbol{\phi}_{s}(t)=\boldsymbol{\Phi}(t)\boldsymbol{\alpha}+\boldsymbol{n}_{s}(t),
\end{equation}
where $\boldsymbol{\alpha}=[\alpha_{1,1},\dots,\alpha_{1,M_{r}},\dots,\alpha_{M_{t},M_{r}}]^{T}$.
Furthermore, define $\boldsymbol{\Psi}=[\boldsymbol{\Phi}^{T}(1),\dots,\boldsymbol{\Phi}^{T}(T)]^{T}\in\mathbb{C}^{NM_{r}T\times M_{t}M_{r}}$
as the concatenated signal formed by $\boldsymbol{\Phi}(t)$ over
total $T$ symbols. Let $\boldsymbol{r}_{\mathrm{S}}\in\mathbb{C}^{NM_{r}T\times1}$
denote the received concatenated signal at all the $M_{r}$ sensing
receivers over $T$ symbols, which is given by
\begin{equation}
\boldsymbol{r}_{\mathrm{S}}=\boldsymbol{\Psi}\boldsymbol{\alpha}+\boldsymbol{n}_{s},\label{eq:receive}
\end{equation}
where $\boldsymbol{n}_{s}\in\mathbb{C}^{NM_{r}T\times1}$ denotes
the noise at the sensing receivers with $\boldsymbol{n}_{s}\sim\mathcal{CN}(\boldsymbol{0},\sigma_{s}^{2}\boldsymbol{I})$.

Let null hypothesis $\mathcal{H}_{0}$ represent that there
is no target in the sensing area, and the alternative hypothesis
$\mathcal{H}_{1}$ represent the existence of the target. As a result,
we formulate the hypothesis as
\begin{equation}
\begin{cases}
\begin{array}{cl}
\mathcal{H}_{0}: & \boldsymbol{r}_{\mathrm{S}}=\boldsymbol{n}_{s},\\
\mathcal{H}_{1}: & \boldsymbol{r}_{\mathrm{S}}=\boldsymbol{\Psi}\boldsymbol{\alpha}+\boldsymbol{n}_{s}.
\end{array}\end{cases}\label{eq:hypo}
\end{equation}
Let $p_{\mathrm{D}}(\{\boldsymbol{w}_{k}\},\boldsymbol{S})$ denote
the detection probability at sensing receivers based on \eqref{eq:hypo}, which is
a function of design variables $\{\boldsymbol{w}_{k}\}$ and $\boldsymbol{S}=\mathbb{E}\big(\boldsymbol{s}^{\mathrm{S}}(t)(\boldsymbol{s}^{\mathrm{S}}(t))^{H}\big)$ to be determined in Section III.

\subsection{Sensing Eavesdropping}

In this subsection, we investigate the sensing eavesdropping model.
Since the transmitted signal $\boldsymbol{x}(t)$ is confidential,
these sensing eavesdroppers are assumed to lack knowledge about $\boldsymbol{x}(t)$,
rendering sensing techniques requiring $\boldsymbol{x}(t)$ inapplicable.
Generally, the sensing eavesdroppers' goal is to detect the transmitted
power from a specific direction to perform target detection. In this
context, we assume that sensing eavesdropper $q\in\mathcal{Q}$ adopts
receive beamforming with a beamforming vector $\boldsymbol{a}_{r}^{H}(\omega_{q})$
to detect the received power along the target direction, similar to
a passive radar technique \cite{6803957}, where $\omega_{q}$ denotes
the AoD from the target to sensing eavesdropper $q\in\mathcal{Q}$.
Let $\boldsymbol{U}_{i,q}\in\mathbb{C}^{N\times N}$ denote the clutter
channel from ISAC transmitter $i\in\mathcal{M}_{t}$ to sensing eavesdropper
$q\in\mathcal{Q}$. We assume that the 
clutter information is available to both the central controller and
sensing eavesdroppers. However, as the sensing eavesdroppers are not aware of the transmitted
signal $\boldsymbol{x}(t)$, they cannot cancel the resulted interference from clutters. Let $\eta_{i,q}\in\mathbb{C}$ represent the complex coefficient from ISAC transmitter
$i\in\mathcal{M}_{t}$ to sensing eavesdropper $q\in\mathcal{Q}$
characterizing the influence of path-loss and target RCS.
Consequently, after applying receive beamformer $\boldsymbol{a}_{r}^{H}(\omega_{q})$, the received signal by sensing eavesdropper
$q\in\mathcal{Q}$ in the $t$-th symbol is given by
\begin{eqnarray}
\tilde{r}_{q}(t) & = & \boldsymbol{a}_{r}^{H}(\omega_{q})\Bigg(\sum_{i=1}^{M_{t}}\eta_{i,q}\boldsymbol{a}_{r}(\omega_{q})\boldsymbol{a}_{t}^{H}(\theta_{i})\boldsymbol{x}_{i}(t)\nonumber \\
 &  & +\underset{\mathrm{clutters}}{\underbrace{\sum_{i=1}^{M_{t}}\boldsymbol{U}_{i,q}^{H}\boldsymbol{x}_{i}(t)}}+\hat{\boldsymbol{n}}_{q}(t)\Bigg),\label{eq:sensing signal}
\end{eqnarray}
where $\hat{\boldsymbol{n}}_{q}(t)\in\mathbb{C}^{N\times1}$ is the
noise with $\hat{\boldsymbol{n}}_{q}(t)\sim\mathcal{CN}(\boldsymbol{0},\sigma_{s}^{2}\boldsymbol{I})$.
For convenience, we rewrite \eqref{eq:sensing signal} as

\begin{equation}
\tilde{r}_{q}(t)=\boldsymbol{a}_{q}^{H}\boldsymbol{x}(t)+\boldsymbol{u}_{q}^{H}\boldsymbol{x}(t)+\boldsymbol{a}_{r}^{H}(\omega_{q})\hat{\boldsymbol{n}}_{q}(t),
\end{equation}
where $\ensuremath{\boldsymbol{a}_{q}=[N\eta_{1,n}\boldsymbol{a}_{t}^{H}(\theta_{1}),\dots,N\eta_{M_{t},n}\boldsymbol{a}_{t}^{H}(\theta_{M_{t}})]^{H}}$
is the equivalent steering vector, and $\boldsymbol{u}_{q}=[\boldsymbol{a}_{r}^{H}(\omega_{q})\boldsymbol{U}_{1,n}^{H},\dots,\boldsymbol{a}_{r}^{H}(\omega_{q})\boldsymbol{U}_{M_{t},n}^{H}]^{H}$
is the equivalent clutter channel. Let $\hat{n}_{q}(t)=\boldsymbol{a}_{r}^{H}(\omega_{q})\hat{\boldsymbol{n}}_{q}(t)\sim\mathcal{CN}(0,N\sigma_{s}^{2})$
denote the equivalent noise. 

Let the null hypothesis $\mathcal{\tilde{H}}_{0}$ represent that there is no target in the sensing area, while the alternative
hypothesis $\mathcal{\tilde{H}}_{1}$ represent the existence of
the target. As a result, we formulate the hypothesis as
\begin{equation}
\begin{cases}
\begin{array}{l}
\mathcal{\tilde{H}}_{0}:\tilde{r}_{q}(t)=\boldsymbol{u}_{q}^{H}\boldsymbol{x}(t)+\hat{n}_{q}(t),\\
\mathcal{\tilde{H}}_{1}:\tilde{r}_{q}(t)=\boldsymbol{a}_{q}^{H}\boldsymbol{x}(t)+\boldsymbol{u}_{q}^{H}\boldsymbol{x}(t)+\hat{n}_{q}(t).
\end{array}\end{cases}\label{eq:hypo-2}
\end{equation}
Let $\ensuremath{\tilde{p}_{q}(\{\boldsymbol{w}_{k}\},\boldsymbol{S})}$
denote the eavesdropping probability at sensing eavesdropper $q\in\mathcal{Q}$
based on \eqref{eq:hypo-2}, which will be derived in Section III
shortly.

\section{Detection and Eavesdropping Probabilities }

Typically, the detection probability $p_{\mathrm{D}}(\{\boldsymbol{w}_{k}\},\boldsymbol{S})$
and the eavesdropping probability $\ensuremath{\tilde{p}_{q}(\{\boldsymbol{w}_{k}\},\boldsymbol{S})}$
are both intricate functions with respect to (w.r.t.) design variables
$\{\boldsymbol{w}_{k}\}$ and $\boldsymbol{S}$, posing tremendous
challenges in both modeling and optimization w.r.t. $\{\boldsymbol{w}_{k}\}$
and $\boldsymbol{S}$. In this section, we establish the relationships
between these probabilities and design variables $\{\boldsymbol{w}_{k}\}$
and $\boldsymbol{S}$. 

\subsection{Detection Probability at Sensing Receivers}

To begin with, we derive the detection probability at the sensing
receivers based on the multi-static sensing model in \eqref{eq:hypo}.
Here, we treat $\boldsymbol{\alpha}$ as the unknown but deterministic
parameters. Note that $\boldsymbol{\alpha}\neq\boldsymbol{0}$ represents
that there is target in the interested area. In this case, the detection
problem can be modeled as a linear model. Following the signal detection
in a noisy scenario \cite{kay2009fundamentals}, the hypothesis in
\eqref{eq:hypo} is equivalently reformulated as 
\begin{equation}
\begin{cases}
\begin{array}{cl}
\mathcal{H}_{0}: & \boldsymbol{\mu}^{T}\boldsymbol{\alpha}=0,\\
\mathcal{H}_{1}: & \boldsymbol{\mu}^{T}\boldsymbol{\alpha}\neq0,
\end{array}\end{cases}\label{eq:hypo-1}
\end{equation}
where $\boldsymbol{\mu}=\bm{1}\in\mathbb{C}^{M_{t}M_{r}\times1}$.
Given that parameters $\boldsymbol{\alpha}$ is unknown and based
on the expression of $\boldsymbol{\Phi}(t)$ in \eqref{eq:Y}, we
proceed to apply the GLRT detector \cite{kay2009fundamentals} based
on the hypothesis in \eqref{eq:hypo-1}. 
\begin{prop}
\textup{The GLRT for \eqref{eq:hypo-1} is given as 
\begin{equation}
\nu(\boldsymbol{Y})=\frac{|\boldsymbol{\mu}^{T}\hat{\boldsymbol{\alpha}}_{1}|^{2}}{\sigma_{s}^{2}\boldsymbol{\mu}^{T}(\boldsymbol{\Psi}^{H}\boldsymbol{\Psi})^{-1}\boldsymbol{\mu}},
\end{equation}
where $\hat{\boldsymbol{\alpha}}_{1}=(\boldsymbol{\Psi}^{H}\boldsymbol{\Psi})^{-1}\boldsymbol{\Psi}^{H}\boldsymbol{r}_{\mathrm{S}}$.
Let $\varXi$ denote the test threshold, the false alarm, and detection
probabilities at sensing receivers are respectively given by 
\begin{eqnarray}
p_{\mathrm{FA}} & = & \Upsilon{}_{\chi^{2}}(\varXi),\\
p_{\mathrm{D}}(\{\boldsymbol{w}_{k}\},\boldsymbol{S}) & = & \Upsilon_{\tilde{\chi}^{2}\big(\lambda(\{\boldsymbol{w}_{k}\},\boldsymbol{S})\big)}(\varXi),
\end{eqnarray}
where $\chi^{2}(\cdot)$ is the central Chi-square distribution, $\tilde{\chi}^{2}(\lambda)$
is the non-central Chi-square distribution with a non-centrality parameter
$\lambda(\{\boldsymbol{w}_{k}\},\boldsymbol{S})=\frac{|\boldsymbol{\mu}^{T}\boldsymbol{\alpha}|^{2}}{\sigma_{s}^{2}\boldsymbol{\mu}^{T}(\boldsymbol{\Psi}^{H}\boldsymbol{\Psi})^{-1}\boldsymbol{\mu}}$,
and $\Upsilon$ represents the right-tailed distribution. Here, $\lambda(\{\boldsymbol{w}_{k}\},\boldsymbol{S})$
is a function of $\boldsymbol{\Psi}$ and thus a function of optimization
variables $\{\boldsymbol{w}_{k}\}$ and $\boldsymbol{S}$.}
\end{prop}
\begin{proof}
Please refer to Appendix A.
\end{proof}

\subsection{Eavesdropping Probability at Sensing Eavesdroppers}

Next, we consider the eavesdropping probability at sensing eavesdropper
$q\in\mathcal{Q}$. Since we assume that the transmitted signal $\boldsymbol{x}(t)$
is confidential, the sensing eavesdropper does not know $\boldsymbol{x}(t)$.
As a result, the sensing eavesdropper can only perform energy detection
to decide whether the target exists, i.e., the detector is given as
$|\tilde{r}_{q}(t)|^{2}$ \cite{wang2023sensing}. Consequently, the
likelihood functions of $y_{q}(t)$ under $\mathcal{\tilde{H}}_{0}$
and $\mathcal{\tilde{H}}_{1}$ are respectively given by
\begin{equation}
\begin{array}{c}
p_{0}(\tilde{r}_{q}(t))=\frac{1}{\pi\zeta_{q}(\{\boldsymbol{w}_{k}\},\boldsymbol{S})}\mathrm{exp}(-\frac{|\tilde{r}_{q}(t)|^{2}}{\zeta_{q}(\{\boldsymbol{w}_{k}\},\boldsymbol{S})}),\\
p_{1}(\tilde{r}_{q}(t))=\frac{1}{\pi\beta_{q}(\{\boldsymbol{w}_{k}\},\boldsymbol{S})}\mathrm{exp}(-\frac{|\tilde{r}_{q}(t)|^{2}}{\beta_{q}(\{\boldsymbol{w}_{k}\},\boldsymbol{S})}),
\end{array}\label{eq:pdf2}
\end{equation}
where we define $\zeta_{q}(\{\boldsymbol{w}_{k}\},\boldsymbol{S})=\boldsymbol{u}_{q}^{H}\boldsymbol{R}\boldsymbol{u}_{q}+N\sigma^{2}$ and 
$\beta_{q}(\{\boldsymbol{w}_{k}\},\boldsymbol{S})=(\boldsymbol{a}_{q}+\boldsymbol{u}_{q})^{H}\boldsymbol{R}(\boldsymbol{a}_{q}+\boldsymbol{u}_{q})+N\sigma^{2}$, with $\boldsymbol{R}$ denoting the covariance matrix in \eqref{eq:cov}.
In accordance with the hypothesis in \eqref{eq:hypo-2}, Proposition
2 provides the maximum target detection probability via the Neyman-Pearson
criterion under the optimal threshold.
\begin{prop}
\textup{The maximum target detection probability under the optimal
threshold is given as 
\begin{equation}
\ensuremath{\tilde{p}_{q}(\{\boldsymbol{w}_{k}\},\boldsymbol{S})}=\big(\frac{\beta_{q}(\{\boldsymbol{w}_{k}\},\boldsymbol{S})}{\zeta_{n}(\{\boldsymbol{w}_{k}\},\boldsymbol{S})}\big)^{-\frac{\zeta_{n}(\{\boldsymbol{w}_{k}\},\boldsymbol{S})}{\beta_{q}(\{\boldsymbol{w}_{k}\},\boldsymbol{S})-\zeta_{n}(\{\boldsymbol{w}_{k}\},\boldsymbol{S})}}.
\end{equation}
}
\end{prop}
\begin{proof}
Please refer to Appendix B. 
\end{proof}
\begin{rem}
It is important to compare the detection probability in Proposition
1 versus the eavesdropping probability in Proposition 2. Notice that sensing receivers are aware of the transmitted signal
$\boldsymbol{x}(t)$, but sensing eavesdroppers do not have such information. As a result, sensing receivers are able to design
a detector via jointly exploiting the knowledge of $\boldsymbol{x}(t)$
in received signal $\boldsymbol{r}_{\mathrm{S}}$, but each
sensing eavesdropping can only perform target detection independently
via detecting the received signal power along the target direction. On the other hand, sensing receivers
have the capability to cancel clutters, while sensing eavesdroppers
do not possess this capability. As such, the inherent clutters result in interference towards sensing eavesdroppers only, which can
be utilized to jam sensing eavesdroppers for preserving the sensing data privacy, which will be discussed
in the subsequent section. 
\end{rem}

\section{Joint Transmit Beamforming Design for Secure Cell-free ISAC}

In this section, we first formulate a joint transmit beamforming problem
for the secure cell-free ISAC system. Our objective is to maximize
the detection probability subject to the minimum SINR constraints at the
CUs. Meanwhile, the formulation considers the maximum SNR constraints
at information eavesdroppers to ensure the confidentiality of information
transmission and the maximum detection probability constraints
at sensing eavesdroppers to maintain the sensing privacy. Following
this, we introduce a three-step approach to achieve a globally optimal
solution.

\subsection{Problem Formulation}

In this subsection, our goal is to maximize the legal detection probability,
as defined in Proposition 1, while ensuring that the CUs meet a minimum
SINR threshold. This is achieved through the optimization of both
the transmit information beamformers $\{\boldsymbol{w}_{k}\}$ and
the dedicated sensing covariance $\boldsymbol{S}$. Furthermore, the
transmission of each ISAC transmitter is constrained by a maximum
power budget $P$. For ease of presentation, let $\boldsymbol{A}_{i}\in\mathbb{C}^{N\times NM_{t}}$
denote an auxiliary binary matrix to extract $\hat{\boldsymbol{w}}_{i,k}$
from $\boldsymbol{w}_{k}$, which is composed of $M_{t}$ individual
$N\times N$ matrices arranged horizontally. Within this composite
matrix, only the $i$-th segment contains an identity matrix, while
all other segments consist of zeros, i.e., 
\begin{equation}
\boldsymbol{A}_{i}=[\boldsymbol{0},\dots,\underset{\textrm{the \ensuremath{i}-th matrx }}{\underbrace{\boldsymbol{I}}},\dots,\boldsymbol{0}].
\end{equation}
In this case, the transmit power constraint at ISAC transmitter $i\in\mathcal{M}_{t}$
is given by
\begin{equation}
\mathbb{E}\big(\boldsymbol{x}_{i}^{H}(t)\boldsymbol{x}_{i}(t)\big)=\mathrm{tr}\Bigg(\boldsymbol{A}_{i}\Big(\sum_{k=1}^{K}\boldsymbol{w}_{k}\boldsymbol{w}_{k}^{H}+\boldsymbol{S}\Big)\boldsymbol{A}_{i}^{H}\Bigg)\leq P.
\end{equation}
Let $\Gamma$, $\Omega$, and $\Lambda$ denote the required minimum
communication SINR threshold for the CUs, the maximum information
eavesdropping SNR threshold for the information eavesdroppers, and
the maximum eavesdropping sensing probability threshold for the sensing
eavesdroppers, respectively. Consequently, the joint transmit beamforming
problem for secure cell-free ISAC is formulated as\begin{subequations}

\begin{eqnarray}
\hspace{-0.8cm}(\mathrm{P1}): & \hspace{-0.2cm}\underset{\{\boldsymbol{w}_{k}\},\boldsymbol{S}}{\max} & \hspace{-0.2cm}\Upsilon_{\tilde{\chi}^{2}\big(\lambda(\{\boldsymbol{w}_{k}\},\boldsymbol{S})\big)}(\varXi)\nonumber \\
 & \hspace{-0.2cm}\mathrm{s.t.} & \hspace{-0.2cm}\gamma_{k}(\{\boldsymbol{w}_{k}\},\boldsymbol{S})\geq\Gamma,\forall k\in\mathcal{K},\label{eq:p1a}\\
 &  & \hspace{-0.2cm}\tilde{\gamma}_{l,k}(\{\boldsymbol{w}_{k}\},\boldsymbol{S})\leq\Omega,\forall l\in\mathcal{L},k\in\mathcal{K},\label{eq:p1b}\\
 &  & \hspace{-0.2cm}\tilde{p}_{q}(\{\boldsymbol{w}_{k}\},\boldsymbol{S})\leq\Lambda,\forall q\in\mathcal{Q},\label{eq:p1c}\\
 &  & \hspace{-0.8cm}\hspace{-0.3cm}\mathrm{tr}\Bigg(\boldsymbol{A}_{i}\Big(\sum_{k=1}^{K}\boldsymbol{w}_{k}\boldsymbol{w}_{k}^{H}+\boldsymbol{S}\Big)\boldsymbol{A}_{i}^{H}\Bigg)\leq P,\forall i\in\mathcal{M}_{t},\label{power}\\
 &  & \hspace{-0.2cm}\boldsymbol{S}\succeq\boldsymbol{0}.\label{eq:semipositive}
\end{eqnarray}
\end{subequations}

It is important to note that problem (P1) exhibits a high degree of
non-convexity, mainly stemming  from the non-convex and non-smooth
nature of the objective function and the constraints presented in
\eqref{eq:p1a}, \eqref{eq:p1b}, and \eqref{eq:p1c}. Therefore,
it is challenging to solve. In the subsequent subsection, we
present an effective approach to address these challenges and obtain
a globally optimal solution to problem (P1).

\subsection{Optimal Solution to Problem (P1)}

In this subsection, we present a three-step approach to obtain a globally
optimal solution to problem (P1). Initially, we simplify the detection
probability and eavesdropping probability to enhance problem tractability.
Subsequently, we reformulate the problem by adopting the SDR method \cite{luo2010semidefinite},
such that the convex SDR version can be optimally solved via off-the-shelf
toolboxes. Finally, we rigorously verify the tightness of the adopted
SDR to ensure the global optimality of the obtained solution.

\subsubsection{Problem Reformulation}

To begin with, we exploit the detection probability expression in
Proposition 1 and establish its equivalent form that is tractable
for optimization. Despite the non-smooth nature of the detection probability
expression in Proposition 1 w.r.t. $\lambda(\{\boldsymbol{w}_{k}\},\boldsymbol{S})$,
it is observed that an increase in the non-centrality parameter $\lambda(\{\boldsymbol{w}_{k}\},\boldsymbol{S})$
results in an increase in the right-tail probability for a non-central
chi-square distribution, given a specific threshold. This is because
that as $\lambda(\{\boldsymbol{w}_{k}\},\boldsymbol{S})$ increases,
the entire non-central chi-square distribution shifts to the right,
leading to an augmentation in the right-tail probability. Consequently,
we establish that maximizing the detection probability
$p_{\mathrm{D}}(\{\boldsymbol{w}_{k}\},\boldsymbol{S})$ is equivalent
to maximizing the non-centrality parameter $\lambda(\{\boldsymbol{w}_{k}\},\boldsymbol{S})$,
which is further equivalent to minimizing $\boldsymbol{\mu}^{T}(\boldsymbol{\Psi}^{H}\boldsymbol{\Psi})^{-1}\boldsymbol{\mu}$
in the denominator. 

Next, we consider the expression $\boldsymbol{\Psi}^{H}\boldsymbol{\Psi}$.
Recall that $\boldsymbol{\Psi}=[\boldsymbol{\Phi}^{T}(1),\dots,\boldsymbol{\Phi}^{T}(T)]^{T}$
and $\boldsymbol{\Phi}(t)=\mathrm{blkdiag}(\boldsymbol{\Phi}_{1}(t),\dots,\boldsymbol{\Phi}_{M_{r}}(t))$.
As such, we derive $\boldsymbol{\Psi}^{H}\boldsymbol{\Psi}$ as
\begin{equation}
\boldsymbol{\Psi}^{H}\boldsymbol{\Psi}=\sum_{t=1}^{T}\mathrm{blkdiag}(\boldsymbol{\Phi}_{1}^{H}(t)\boldsymbol{\Phi}_{1}(t),\dots,\boldsymbol{\Phi}_{M_{r}}^{H}(t)\boldsymbol{\Phi}_{M_{r}}(t)).\label{eq:eq1}
\end{equation}
Notice that the correlation of the received signal $\boldsymbol{\Phi}_{j}^{H}(t)\boldsymbol{\Phi}_{j}(t)$
at sensing receiver $j\in\mathcal{M}_{r}$ is given by
\begin{equation}
\boldsymbol{\Phi}_{j}^{H}(t)\boldsymbol{\Phi}_{j}(t)\hspace{-2pt}=\hspace{-2pt}\hspace{-2pt}\left[\hspace{-2pt}\begin{array}{ccc}
\boldsymbol{\phi}_{1,j}^{H}(t)\boldsymbol{\phi}_{1,j}(t) & \hspace{-2pt}\hspace{-2pt}\hspace{-2pt}\hspace{-2pt}\dots\hspace{-2pt}\hspace{-2pt}\hspace{-2pt}\hspace{-2pt} & \boldsymbol{\phi}_{1,j}^{H}(t)\boldsymbol{\phi}_{M_{t},j}(t)\\
\dots & \hspace{-2pt}\hspace{-2pt}\hspace{-2pt}\hspace{-2pt}\dots\hspace{-2pt}\hspace{-2pt}\hspace{-2pt}\hspace{-2pt} & \dots\\
\boldsymbol{\phi}_{M_{t},j}^{H}(t)\boldsymbol{\phi}_{1,j}(t) & \hspace{-2pt}\hspace{-2pt}\hspace{-2pt}\hspace{-2pt}\dots\hspace{-2pt}\hspace{-2pt}\hspace{-2pt}\hspace{-2pt} & \boldsymbol{\phi}_{M_{t},j}^{H}(t)\boldsymbol{\phi}_{M_{t},j}(t)
\end{array}\hspace{-2pt}\hspace{-2pt}\hspace{-2pt}\right].\label{eq:eq2}
\end{equation}
Also notice that for the received signals from ISAC transmitters $m\in\mathcal{M}_{t}$
and $n\in\mathcal{M}_{t}$ at sensing receiver $j\in\mathcal{M}_{r}$,
it follows that 
\begin{eqnarray}
\sum_{t=1}^{T}\boldsymbol{\phi}_{m,j}^{H}(t)\boldsymbol{\phi}_{n,j}(t) & = & N\sum_{t=1}^{T}\boldsymbol{x}_{m}^{H}(t)\boldsymbol{a}_{t}(\theta_{m})\boldsymbol{a}_{t}^{H}(\theta_{n})\boldsymbol{x}_{n}(t)\nonumber \\
 & \overset{(\mathrm{a})}{=} & TN\mathrm{tr}\big(\boldsymbol{a}_{t}(\theta_{m})\boldsymbol{a}_{t}^{H}(\theta_{n})\boldsymbol{R}_{m,n}\big),\label{eq:eq3}
\end{eqnarray}
where $\boldsymbol{R}_{m,n}=\mathrm{\mathbb{E}}\big(\boldsymbol{x}_{n}(t)\boldsymbol{x}_{m}^{H}(t)\big)$
and equality (a) holds due to the fact the statistical covariance
matrix $\boldsymbol{R}_{m,n}$ is equivalent to the sample covariance matrix 
when the number $T$ of symbols is sufficiently large \cite{liu2021integrated}.
As a result, by combining \eqref{eq:eq1}, \eqref{eq:eq2}, and \eqref{eq:eq3},
we rewrite $\boldsymbol{\Psi}^{H}\boldsymbol{\Psi}$ as 
\begin{eqnarray}
\boldsymbol{\Psi}^{H}\boldsymbol{\Psi} & = & \mathrm{blkdiag}(\underset{M_{r}}{\underbrace{\boldsymbol{R}_{\Phi},\dots,\boldsymbol{R}_{\Phi}}}),
\end{eqnarray}
where $\boldsymbol{R}_{\Phi}\in\mathbb{C}^{M_{t}\times M_{t}}$ and
\begin{equation}
\boldsymbol{R}_{\Phi}[m,n]=TN\mathrm{tr}\big(\boldsymbol{a}_{t}(\theta_{m})\boldsymbol{a}_{t}^{H}(\theta_{n})\boldsymbol{A}_{m}\boldsymbol{R}\boldsymbol{A}_{n}^{H}\big).\label{eq:R_y}
\end{equation}
As a result, maximizing the detection
probability expression in Proposition 1 is shown to be equivalent to minimizing the following tractable and equivalent form\footnote{It is insightful to discuss the covariance matrix $\boldsymbol{R}_{\Phi}$
in \eqref{eq:R_y}. The diagonal elements in $\boldsymbol{R}_{\Phi}$
represent the transmitted power directed towards the sensing target
at ISAC transmitters. Meanwhile, the non-diagonal elements in $\boldsymbol{R}_{\Phi}$
capture the covariance of transmitted signals from different ISAC
transmitters. Notably, the covariance of transmitted signals from
various ISAC transmitters should be properly designed together with
the power directed towards the target to enhance detection performance.}:
\begin{equation}
\boldsymbol{\mu}^{T}(\boldsymbol{\Psi}^{H}\boldsymbol{\Psi})^{-1}\boldsymbol{\mu}.\label{eq:detection prob}
\end{equation}

Next, we address the non-convex eavesdropping probability in constraints
\eqref{eq:p1c}, which is given in the following form:
\begin{equation}
\ensuremath{\tilde{p}_{q}(\{\boldsymbol{w}_{k}\},\boldsymbol{S})}=\big(\frac{\beta_{q}(\{\boldsymbol{w}_{k}\},\boldsymbol{S})}{\zeta_{n}(\{\boldsymbol{w}_{k}\},\boldsymbol{S})}\big)^{-\frac{1}{\frac{\beta_{q}(\{\boldsymbol{w}_{k}\},\boldsymbol{S})}{\zeta_{n}(\{\boldsymbol{w}_{k}\},\boldsymbol{S})}-1}}.
\end{equation}
Notice that function $f(x)=x^{-\frac{1}{x-1}}$ 
increases monotonically for $0<x<1$ and $x>1$. Therefore, by letting
$\Gamma_{\mathrm{d}}$ denote the solution to the equation
$f(x)=\Lambda$, we equivalently rewrite the sensing eavesdropping constraints
in \eqref{eq:p1c} as\footnote{For the eavesdropping probability, a larger denominator $\zeta_{n}(\{\boldsymbol{w}_{k}\},\boldsymbol{S})$
at \eqref{eq:sensing eavesdropper} contributes to a decrease in the
sensing eavesdropping detection probability. This denominator $\zeta_{n}(\{\boldsymbol{w}_{k}\},\boldsymbol{S})$
characterizes the received clutter plus noise power at the sensing eavesdroppers.
Fundamentally, this indicates an opportunity to exploit clutter information
to jam the sensing eavesdroppers. Conversely, a smaller numerator
$\beta_{q}(\{\boldsymbol{w}_{k}\},\boldsymbol{S})$ represents the
ability to decrease the received power at the sensing receivers when
a target exists. This intentional reduction in received power aims
to prevent the sensing eavesdroppers from detecting the presence of
a target. } 
\begin{equation}
\frac{\beta_{q}(\{\boldsymbol{w}_{k}\},\boldsymbol{S})}{\zeta_{n}(\{\boldsymbol{w}_{k}\},\boldsymbol{S})}\leq\Gamma_{\mathrm{d}},\forall q\in\mathcal{Q}.\label{eq:sensing eavesdropper}
\end{equation}

Finally, we reformulate the SINR constraints in \eqref{eq:p1a} as
\begin{equation}
\boldsymbol{h}_{k}^{H}\Big(\boldsymbol{w}_{k}\boldsymbol{w}_{k}^{H}-\Gamma(\sum_{i\neq k}^{K}\boldsymbol{w}_{i}\boldsymbol{w}_{i}^{H}+\boldsymbol{S})\Big)\boldsymbol{h}_{k}\geq\Gamma\sigma^{2},\forall k\in\mathcal{K}.\label{eq:SINR1}
\end{equation}
Similarly, we equivalently rewrite the information eavesdropping SNR
constraints in \eqref{eq:p1b} as 
\begin{equation}
\boldsymbol{g}_{l}^{H}\boldsymbol{w}_{k}\boldsymbol{w}_{k}^{H}\boldsymbol{g}_{l}\leq\Omega\sigma^{2},\forall l\in\mathcal{L},\forall k\in\mathcal{K}.\label{eq:SNR1}
\end{equation}

By combining \eqref{eq:detection prob}, \eqref{eq:sensing eavesdropper},
\eqref{eq:SINR1}, and \eqref{eq:SNR1}, problem (P1) is equivalently
reformulated as \vspace{-0.2cm}
\begin{eqnarray*}
(\mathrm{P1.1}): & \underset{\{\boldsymbol{w}_{k}\},\boldsymbol{S}}{\min} & \boldsymbol{\mu}^{T}(\boldsymbol{\Psi}^{H}\boldsymbol{\Psi})^{-1}\boldsymbol{\mu}\\
 & \mathrm{s.t.} & \textrm{\eqref{power}, \eqref{eq:semipositive}, \eqref{eq:sensing eavesdropper}, \eqref{eq:SINR1}, and \eqref{eq:SNR1}.}
\end{eqnarray*}
It is noting that problem (P1.1) is still a non-convex problem due
to the non-convexity of the objective function and the constraints
in \eqref{eq:sensing eavesdropper} and \eqref{eq:SINR1}.

\subsubsection{SDR-Based Solution to Problem (P1.1)}

In the following, we present a SDR-based approach to address problem
(P1.1). By introducing auxiliary variables $\boldsymbol{W}_{k}=\boldsymbol{w}_{k}\boldsymbol{w}_{k}^{H}\succeq\boldsymbol{0}$
with $\textrm{rank}(\boldsymbol{W}_{k})\leq1$, $\forall k\in\mathcal{K}$,
we equivalently reformulate the SINR constraints in \eqref{eq:SINR1}
as 
\begin{equation}
\boldsymbol{h}_{k}^{H}\big(\boldsymbol{W}_{k}-\Gamma(\sum_{i\neq k}^{K}\boldsymbol{W}_{i}+\boldsymbol{S})\big)\boldsymbol{h}_{k}\geq\Gamma\sigma^{2},\forall k\in\mathcal{K}.\label{eq:SINR2}
\end{equation}
Similarly, we equivalently reformulate \eqref{eq:SNR1} as 
\begin{equation}
\boldsymbol{g}_{l}^{H}\boldsymbol{W}_{k}\boldsymbol{g}_{l}\leq\Omega\sigma^{2},\forall l\in\mathcal{L},\forall k\in\mathcal{K}.\label{eq:SNR2}
\end{equation}

Subsequently, we aim to address the sensing eavesdropping constraints
in \eqref{eq:sensing eavesdropper}. Based on the expression of $\beta_{q}(\{\boldsymbol{w}_{k}\},\boldsymbol{S})$
and $\zeta_{n}(\{\boldsymbol{w}_{k}\},\boldsymbol{S})$ in Proposition
2, we equivalently reformulate \eqref{eq:sensing eavesdropper}
as 
\begin{equation}
\boldsymbol{u}_{q}^{H}\boldsymbol{R}\boldsymbol{u}_{q}+N\sigma^{2}\leq\Gamma_{\mathrm{d}}\big((\boldsymbol{a}_{q}+\boldsymbol{u}_{q})^{H}\boldsymbol{R}(\boldsymbol{a}_{q}+\boldsymbol{u}_{q})+N\sigma^{2}\big),\forall q\in\mathcal{Q}.\label{eq:sensing eavesdropping2}
\end{equation}
Furthermore, the transmit power constraints at ISAC transmitters in \eqref{power}
are equivalently reformulated as 
\begin{equation}
\mathrm{tr}\Bigg(\boldsymbol{A}_{i}\Big(\sum_{k=1}^{K}\boldsymbol{W}_{k}+\boldsymbol{S}\Big)\boldsymbol{A}_{i}^{H}\Bigg)\leq P,\forall i\in\mathcal{M}_{t}.\label{eq:power2-1}
\end{equation}
Moreover, recall that $\boldsymbol{\Psi}^{H}\boldsymbol{\Psi}$ is
affine w.r.t. $\boldsymbol{R} = \sum_{k=1}^{K}\boldsymbol{W}_{k}+\boldsymbol{S}$ and thus equivalently affine w.r.t. $\{\boldsymbol{W}_{k}\}$ and $\boldsymbol{S}$.
Consequently, we equivalently reformulate problem (P1.1) as \begin{subequations}
\begin{eqnarray}
\hspace{-1cm}(\mathrm{P1.2}): & \underset{\{\boldsymbol{W}_{k}\},\boldsymbol{S}}{\min} & \boldsymbol{\mu}^{T}(\boldsymbol{\Psi}^{H}\boldsymbol{\Psi})^{-1}\boldsymbol{\mu}\nonumber \\
 & \mathrm{s.t.} & \textrm{\eqref{eq:SINR2}, \eqref{eq:SNR2}, \eqref{eq:sensing eavesdropping2}, \eqref{eq:power2-1}, and \eqref{eq:semipositive}. }\nonumber \\
 &  & \boldsymbol{W}_{k}\succeq\boldsymbol{0},\textrm{rank}(\boldsymbol{W}_{k})\leq1,\forall k\in\mathcal{K}.\label{eq:rank}
\end{eqnarray}
\end{subequations} Notice that problem (P1.2) is still non-convex
due to the non-convex rank constraints in \eqref{eq:rank}. Then we
drop the rank constraints in \eqref{eq:rank} and obtain the SDR version
of problem (P1.2) as
\begin{eqnarray*}
(\textrm{SDR-1.2}): & \underset{\{\boldsymbol{W}_{k}\},\boldsymbol{S}}{\min} & \boldsymbol{\mu}^{T}(\boldsymbol{\Psi}^{H}\boldsymbol{\Psi})^{-1}\boldsymbol{\mu}\\
 & \mathrm{s.t.} & \textrm{\eqref{eq:SINR2}, \eqref{eq:SNR2}, \eqref{eq:sensing eavesdropping2}, \eqref{eq:power2-1}, and \eqref{eq:semipositive}. }
\end{eqnarray*}
In this context, $\textrm{\eqref{eq:SINR2}, \eqref{eq:SNR2}, \eqref{eq:sensing eavesdropping2}, \eqref{eq:power2-1}, and \eqref{eq:semipositive}}$
are all affine constraints w.r.t. $\{\boldsymbol{W}_{k}\}$ and $\boldsymbol{S}$.
Also notice that the inverse of a semidefinite matrix is a strict
convex function \cite{horn2012matrix}. As a result, the objective
function $\boldsymbol{\mu}^{T}(\boldsymbol{\Psi}^{H}\boldsymbol{\Psi})^{-1}\boldsymbol{\mu}$
is a strict convex function w.r.t. $\{\boldsymbol{W}_{k}\}$ and $\boldsymbol{S}$. Therefore, problem (SDR-1.2) is a convex optimization problem that is solvable
via an off-the-shelf toolbox, such as CVX \cite{grant2014cvx}. Let
$\{\boldsymbol{W}_{k}^{*}\}$ and $\boldsymbol{S}^{*}$ denote the
optimal solution to problem (SDR-1.2) and $\boldsymbol{R}^{*}=\sum_{k=1}^{K}\boldsymbol{W}_{k}^{*}+\boldsymbol{S}^{*}$. 

\subsubsection{Tightness of SDR}
In general, the optimal solution $\{\boldsymbol{W}_{k}^{*}\}$ and
$\boldsymbol{S}^{*}$ may not satisfy the rank constraints in \eqref{eq:rank} and
thus are not necessarily optimal to the original problem (P1.2). In
the following, we rigorously show the tightness of SDR in the following.

%by obtaining
%the optimal solution $\{\boldsymbol{W}_{k}^{\mathrm{opt}}\}$ and
%$\boldsymbol{S}^{\mathrm{opt}}$ that satisfy $\textrm{rank}(\boldsymbol{W}_{k})=1$,
%$\forall k\in\mathcal{K}$ and achieve the same optimal value. 
\begin{prop}
\textup{Based on the obtained solution $\{\boldsymbol{W}_{k}^{*}\}$
and $\boldsymbol{S}^{*}$, we can always construct an alternative solution as
\begin{eqnarray}
 &  & \boldsymbol{W}_{k}^{\mathrm{opt}}=\frac{\boldsymbol{W}_{k}^{*}\boldsymbol{h}_{k}\boldsymbol{h}_{k}^{H}\boldsymbol{W}_{k}^{*}}{\boldsymbol{h}_{k}^{H}\boldsymbol{W}_{k}^{*}\boldsymbol{h}_{k}},\forall k\in\mathcal{K},\\
 &  & \boldsymbol{S}^{\mathrm{opt}}=\boldsymbol{R}^{*}-\sum_{k=1}^{K}\boldsymbol{W}_{k}^{\mathrm{opt}},
\end{eqnarray}
with $\textrm{rank}(\boldsymbol{W}_{k}^{\mathrm{opt}})=1$,
$\forall k\in\mathcal{K}$.} 
\textup{The constructed solution $\{\boldsymbol{W}_{k}^{\mathrm{opt}}\}$
and $\boldsymbol{S}^{\mathrm{opt}}$ achieves the same optimal objective
function achieved by solution $\{\boldsymbol{W}_{k}^{*}\}$ and $\boldsymbol{S}^{*}$
while satisfying the rank constraints in \eqref{eq:rank}. Therefore,
the solution of $\{\boldsymbol{W}_{k}^{\mathrm{opt}}\}$ and $\boldsymbol{S}^{\mathrm{opt}}$
is globally optimal to problem (P1.2).}
\end{prop}
\begin{proof}
Please refer to Appendix C.
\end{proof}
Based on Proposition 3, the optimal solution to the joint secure transmit
beamforming problem (P1) is finally obtained.

\section{Alternative Secure ISAC Designs}

In this section, we propose two alternative transmit beamforming designs
for secure ISAC. The first alternative design focuses on maximizing
the sensing SNR rather than directly maximizing detection probability.
The second design involves coordinated transmit beamforming in a multi-cell
ISAC configuration rather than cell-free transmission.

\subsection{Sensing SNR Maximization}

In this subsection, we consider the sensing SNR maximization problem
in the secure cell-free ISAC scenario, which is a widely adopted intuitive design objective  \cite{demirhan2023cell}. Following the discussion in
\cite{demirhan2023cell}, the sensing SNR is defined as the ratio
of the sum power at all sensing receivers to the sum noise power.
The received power from transmitter $i\in M_{t}$ to the target direction
is given as $\boldsymbol{a}_{t}^{H}(\theta_{i})\boldsymbol{A}_{i}\boldsymbol{R}\boldsymbol{A}_{i}^{H}\boldsymbol{a}_{t}(\theta_{i})$.
Consequently, the received sensing SNR is given as \vspace{-0.2cm}
\begin{equation}
\gamma_{\mathrm{S}}(\{\boldsymbol{w}_{k}\},\boldsymbol{S})=\frac{\sum_{i=1}^{M_{t}}\boldsymbol{a}_{t}^{H}(\theta_{i})\boldsymbol{A}_{i}\boldsymbol{R}\boldsymbol{A}_{i}^{H}\boldsymbol{a}_{t}(\theta_{i})}{M_{t}\sigma_{s}^{2}},\label{eq:sensing SNR}
\end{equation}
where ${\boldsymbol{R}}=\mathbb{E}\big(\boldsymbol{x}(t)\boldsymbol{x}^{H}(t)\big)$
is transmit covariance of the $M_{t}$ ISAC transmitters. Consequently,
we formulate a sensing SNR maximization problem within our secure
cell-free ISAC system as follows, by replacing the objective function
in (P1) as the sensing SNR in \eqref{eq:sensing SNR}.\begin{subequations}

\vspace{-0.2cm}
\begin{eqnarray*}
(\mathrm{P2}): & \underset{\{\boldsymbol{w}_{k}\},\boldsymbol{S}}{\max} & \frac{\sum_{i=1}^{M_{t}}\boldsymbol{a}_{t}^{H}(\theta_{i})\boldsymbol{A}_{i}\boldsymbol{R}\boldsymbol{A}_{i}^{H}\boldsymbol{a}_{t}(\theta_{i})}{M_{t}\sigma_{s}^{2}}\\
 & \mathrm{s.t.} & \textrm{\textrm{\eqref{eq:p1a}, \eqref{eq:p1b}, \eqref{eq:p1c}, \eqref{power}}, and \eqref{eq:semipositive}.}
\end{eqnarray*}
\end{subequations}

Based on the similar problem reformulation and SDR technique for problem
(P1), we equivalently reformulate problem (P2) as \begin{subequations}

\vspace{-0.2cm}
\begin{eqnarray}
\hspace{-1cm}(\mathrm{P2.1}): & \underset{\{\boldsymbol{W}_{k}\},\boldsymbol{S}}{\min} & \sum_{i=1}^{M_{t}}\boldsymbol{a}_{t}^{H}(\theta_{i})\boldsymbol{A}_{i}\boldsymbol{R}\boldsymbol{A}_{i}^{H}\boldsymbol{a}_{t}(\theta_{i})\nonumber \\
 & \mathrm{s.t.} & \textrm{\eqref{eq:SINR2}, \eqref{eq:SNR2}, \eqref{eq:sensing eavesdropping2}, \eqref{eq:power2-1}, and \eqref{eq:semipositive}. }\nonumber \\
 &  & \boldsymbol{W}_{k}\succeq\boldsymbol{0},\textrm{rank}(\boldsymbol{W}_{k})\leq1,\forall k\in\mathcal{K}.\label{eq:rank-1}
\end{eqnarray}
\end{subequations} Then, we directly drop the rank constraints in
\eqref{eq:rank-1} and obtain the SDR version of problem (P2.1) as\vspace{-0.2cm}
\begin{eqnarray*}
(\textrm{SDR-2.1}): & \underset{\{\boldsymbol{W}_{k}\},\boldsymbol{S}}{\min} & \sum_{i=1}^{M_{t}}\boldsymbol{a}_{t}^{H}(\theta_{i})\boldsymbol{A}_{i}\boldsymbol{R}\boldsymbol{A}_{i}^{H}\boldsymbol{a}_{t}(\theta_{i})\\
 & \mathrm{s.t.} & \textrm{\eqref{eq:SINR2}, \eqref{eq:SNR2}, \eqref{eq:sensing eavesdropping2}, \eqref{eq:power2-1}, and \eqref{eq:semipositive}. }
\end{eqnarray*}
Due to the fact that the objective function is affine
w.r.t. $\{\boldsymbol{W}_{k}\}$ and $\boldsymbol{S}$, problem (SDR-2.1)
is a convex optimization problem that is solvable via an off-the-shelf
toolbox, such as CVX \cite{grant2014cvx}. It can be similarly verified
that the construction in Proposition 3 is also applicable here to
obtain an optimal solution to problem (P2.1). Thus, problem (P2.1)
can also be optimally solved.

\subsection{Coordinated Transmit Beamforming Design}

In this subsection, we propose another alternative design based on
coordinated beamforming. In the coordinated beamforming scenario,
each CU is served only by one ISAC transmitter and the signals from
different ISAC transmitters are independent \cite{cheng2023optimal}.
Let $\mathcal{K}_{i}$ denote the set of CUs associated to ISAC transmitter
$i\in\mathcal{M}_{t}$. Equivalently, the transmit beamforming vector
at ISAC transmitter $i\in\mathcal{M}_{t}$ for CU $k\notin\mathcal{K}_{i}$
is forced to be $\boldsymbol{0}$, i.e., \vspace{-0.2cm}
\begin{equation}
\hat{\boldsymbol{w}}_{i,k}=\boldsymbol{0},\forall i\in\mathcal{M}_{t},k\notin\mathcal{K}_{i}.\label{eq:coordinated}
\end{equation}
Then, we exploit the binary auxiliary matrix $\boldsymbol{A}_{i}$
to transform \eqref{eq:coordinated} as
\begin{equation}
\boldsymbol{A}_{i}\boldsymbol{w}_{k}=\boldsymbol{0},\forall i\in\mathcal{M}_{t},k\notin\mathcal{K}_{i}.\label{eq:coordinated2}
\end{equation}
 As a result, the detection probability maximization problem for coordinated transmit
beamforming in the secure multi-cell ISAC system is formulated as
\begin{subequations}
\begin{eqnarray*}
(\mathrm{P3}): & \underset{\{\boldsymbol{w}_{k}\},\boldsymbol{S}}{\max} & \Upsilon_{\tilde{\chi}^{2}\big(\lambda(\{\boldsymbol{w}_{k}\},\boldsymbol{S})\big)}(\varPsi)\\
 & \mathrm{s.t.} & \textrm{\textrm{\eqref{eq:coordinated2}, \eqref{eq:p1b}, \eqref{eq:p1c}, \eqref{power}}, and \eqref{eq:semipositive}.}
\end{eqnarray*}
\end{subequations}

In the following, we present the optimal solution to problem (P3).
To begin with, we utilize the same problem reformulation and SDR technique
as in Section IV for problem (P3). Based on the constraints in \eqref{eq:coordinated2},
the auxiliary variables $\boldsymbol{W}_{k}\in\mathbb{C}^{NM_{t}\times NM_{t}}$
should maintain the structure 
\begin{equation}
\boldsymbol{W}_{k}=\mathrm{blkdiag}(\boldsymbol{0},\dots,\underset{\textrm{The \ensuremath{i}-th matrix }}{\underbrace{\hat{\boldsymbol{w}}_{i,k}\hat{\boldsymbol{w}}_{i,k}^{H}}},\dots,\boldsymbol{0}),k\in\mathcal{K}_{i}.\label{eq:str}
\end{equation}
We introduce new auxiliary variables $\hat{\boldsymbol{W}}_{k}\in\mathbb{C}^{N\times N}$
and formulate the structure in \eqref{eq:str} as 
\begin{equation}
\boldsymbol{W}_{k}=\boldsymbol{J}_{i}\otimes\hat{\boldsymbol{W}}_{k},\forall k\in\mathcal{K},\label{eq:coordinated3}
\end{equation}
where $\boldsymbol{J}_{i}=\mathrm{diag}(0,\dots,\underset{\textrm{The \ensuremath{i}-th element}}{\underbrace{1}},\dots,0)$,
$k\in\mathcal{K}_{i}$. As a result, problem (P3) is reformulated
as \begin{subequations}
\begin{eqnarray}
\hspace{-0.5cm}(\mathrm{P3.1}): & \hspace{-0.5cm}\underset{\{\boldsymbol{W}_{i,k},\boldsymbol{W}_{k}\},\boldsymbol{S}}{\min} & \boldsymbol{\mu}^{T}(\boldsymbol{\Psi}^{H}\boldsymbol{\Psi})^{-1}\boldsymbol{\mu}\nonumber \\
 & \mathrm{s.t.} & \textrm{\eqref{eq:coordinated3}, \eqref{eq:SINR2}, \eqref{eq:SNR2}, \eqref{eq:sensing eavesdropping2}, \eqref{eq:power2-1}, and \eqref{eq:semipositive} }\nonumber \\
 &  & \boldsymbol{W}_{k}\succeq\boldsymbol{0},\textrm{rank}(\boldsymbol{W}_{k})\leq1,\forall k\in\mathcal{K}.\label{eq:rank-1-1}
\end{eqnarray}
\end{subequations}Then, we drop the rank constraints in \eqref{eq:rank-1-1}
to obtain the SDR version of problem (P3.1):
\begin{eqnarray*}
(\textrm{SDR-3.1}): & \hspace{-0.5cm}\underset{\{\boldsymbol{W}_{k}\},\boldsymbol{S},\{\boldsymbol{W}_{i,k}\}}{\min} & \boldsymbol{\mu}^{T}(\boldsymbol{Y}^{H}\boldsymbol{Y})^{-1}\boldsymbol{\mu}\\
 & \mathrm{s.t.} & \hspace{-0.8cm}\textrm{\eqref{eq:coordinated3}, \eqref{eq:SINR2}, \eqref{eq:SNR2}, \eqref{eq:sensing eavesdropping2}, \eqref{eq:power2-1}, and \eqref{eq:semipositive}. }
\end{eqnarray*}
Problem (SDR-3.1) is a convex optimization problem that is solvable
via CVX. Let $\{\boldsymbol{W}_{k}^{\star}\}$, $\boldsymbol{S}^{\star}$,
and $\{\hat{\boldsymbol{W}}_{k}^{\star}\}$ denote the obtained solution
and $\boldsymbol{R}^{\star}=\sum_{k=1}^{K}\boldsymbol{W}_{k}^{\star}+\boldsymbol{S}^{\star}$.
In the following, we rigorously show the tightness of SDR. 

%by obtaining
%the optimal solution $\{\bar{\boldsymbol{W}}_{k}^{\mathrm{opt}}\}$
%and $\bar{\boldsymbol{S}}^{\mathrm{opt}}$, and $\{\bar{\boldsymbol{W}}_{k}^{\mathrm{opt}}\}$
%that satisfy $\textrm{rank}(\boldsymbol{W}_{k})\leq1$, $\forall k\in\mathcal{K}$
%and achieve the same optimal value. 

\begin{prop}
\textup{Based on the obtained solution }$\{\boldsymbol{W}_{k}^{\star}\}$,
$\boldsymbol{S}^{\star}$, \textup{and} $\{\hat{\boldsymbol{W}}_{k}^{\star}\}$\textup{,
we construct 
\begin{eqnarray}
 &  & \hat{\boldsymbol{W}}_{k}^{\mathrm{opt}}=\frac{\hat{\boldsymbol{W}}_{k}^{\star}\boldsymbol{h}_{i,k}\boldsymbol{h}_{i,k}^{H}\hat{\boldsymbol{W}}_{k}^{\star}}{\boldsymbol{h}_{i,k}^{H}\hat{\boldsymbol{W}}_{k}^{\star}\boldsymbol{h}_{i,k}},\\
 &  & \bar{\boldsymbol{W}}_{k}^{\mathrm{opt}}=\boldsymbol{J}_{i}\otimes\hat{\boldsymbol{W}}_{k}^{\mathrm{opt}},\\
 &  & \bar{\boldsymbol{S}}^{\mathrm{opt}}=\boldsymbol{R}^{\star}-\sum_{k=1}^{K}\bar{\boldsymbol{W}}_{k}^{\mathrm{opt}},
\end{eqnarray}
with $\textrm{rank}(\hat{\boldsymbol{W}}_{k}^{\mathrm{opt}})\leq1$, $\forall k\in\mathcal{K}$.
As a result, the constructed solution$\{\bar{\boldsymbol{W}}_{k}^{\mathrm{opt}}\}$
and $\bar{\boldsymbol{S}}^{\mathrm{opt}}$, and $\{\hat{\boldsymbol{W}}_{k}^{\mathrm{opt}}\}$
achieves the same optimal objective function and satisfies the rank
constraints in \eqref{eq:rank}, which is thus  optimal to problem (P3).}
\end{prop}
\begin{proof}
This proposition can be similarly proved as for Proposition 3. Therefore,
the details are omitted. 
\end{proof}
\begin{rem}
It is crucial to highlight
that the objective of maximizing the sensing SNR in a cell-free system
is equivalent to maximizing the trace of our covariance matrix in
\eqref{eq:R_y}. While this design prioritizes maximizing the transmitted
power toward the target direction, it may overlook the significance
of non-diagonal elements in $\boldsymbol{R}_{\Phi}$. Indeed, these
non-diagonal elements represent the covariance of signals from different
ISAC transmitters, indicating that the sensing SNR maximization design
might not fully leverage the joint signal processing gain in the cell-free
network to maximize the detection probability, and thus may lead to
sub-optimal performance in general. Next, the coordinated beamforming
 neglects the importance of the covariance of signals from different
ISAC transmitters by allowing the independent signal
transmission at ISAC transmitters for achieving lower implementation complexity. This, however, may lead to sub-optimal performance as
compared to the optimal cell-free architecture. The arguments will be validated in Section VI. 
\end{rem}

\section{Numerical Results}

\begin{figure}[H]
\centering\includegraphics[scale=0.4]{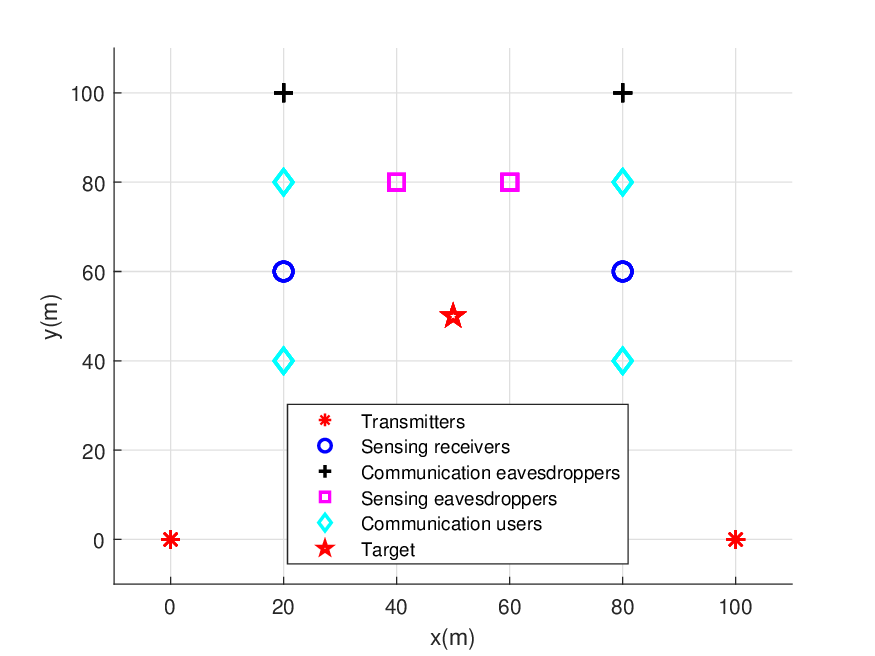}\caption{Simulation topology of the secure cell-free ISAC system.}
\end{figure}

In this section, we provide numerical results to validate the performance
of our proposed joint beamforming solution for the secure cell-free ISAC
system. If not explicitly specified, the topology of our considered
cell-free ISAC system is illustrated in Fig. 2. Within a $100\textrm{ m}\times100\textrm{ m}$
field, we position $M_{t}=2$ ISAC transmitters at $[0\textrm{ m},0\textrm{ m}]$
and $[100\textrm{ m},0\textrm{ m}]$, and $M_{r}=2$ receivers at
$[20\textrm{ m},60\textrm{ m}]$ and $[80\textrm{ m},60\textrm{ m}]$,
respectively. The target of interest is situated at $[50\textrm{ m},50\textrm{ m}]$.
Additionally, there are $K=4$ CUs located at $[20\textrm{ m},40\textrm{ m}]$,
$[20\textrm{ m},80\textrm{ m}]$, $[80\textrm{ m},40\textrm{ m}]$,
and $[80\textrm{ m},80\textrm{ m}]$, $L=2$ information eavesdroppers positioned at $[20\textrm{ m},100\textrm{ m}]$ and $[80\textrm{ m},100\textrm{ m}]$,
as well as $Q=2$ sensing eavesdroppers placed at $[40\textrm{ m},80\textrm{ m}]$
and $[60\textrm{ m},80\textrm{ m}]$. The path-loss at a reference
distance of $\textrm{1 m}$ is $-40\textrm{ dB}$, while the path-loss
exponent is standardized at 3. The noise power is set as $\sigma^{2}=\sigma_{s}^{2}=-100\textrm{ dBm}$.
The reflection coefficient w.r.t. RCS is a Gaussian random variable
with zero mean and variance $0.01$. The communication channel between
each ISAC transmitter and the CUs or information eavesdroppers is modeled
using Rician fading with a Rician factor of $K_{r}=5\textrm{ dB}$.
The clutter channel of sensing eavesdroppers is generated by two randomly
located scatterers in this area. Furthermore, the communication SINR
threshold is established at $\Gamma=10\textrm{ dB}$, and the information
eavesdropping SNR threshold is set as $\Omega=5\textrm{ dB}$. Moreover,
the false alarm probability is standardized at 0.05, while the threshold
for eavesdropping probability is set as $\Lambda=0.4$. 

For the purpose of comparison, we also illustrate an upper bound on the detection probability through
a sensing only scheme, where only the transmit power constraints in
problem (P1) are taken into account. This sensing only scheme serves
as an upper limit in the cell-free ISAC system.

\begin{figure}[H]
\centering\includegraphics[scale=0.45]{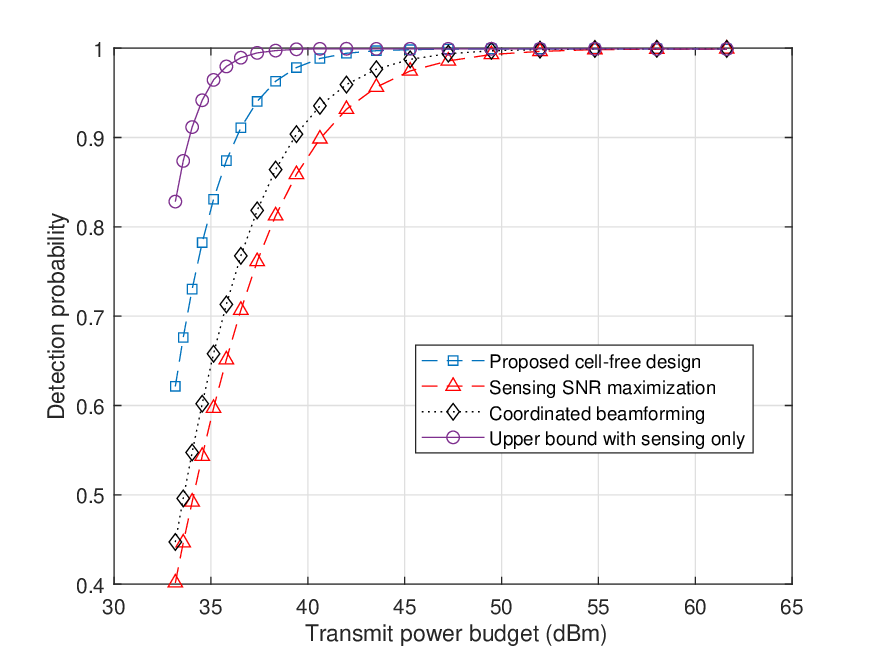}\centering\caption{\label{fig:3}The detection probability versus transmit power budget
$P$.}
\end{figure}

Fig. \ref{fig:3} shows the detection probability versus the transmit
power budget $P$ at each transmitter. Notably, the sensing only
scheme serves as a performance upper bound for maximizing the detection
probability in this scenario. Furthermore, the two benchmark schemes
exhibit similar performance and maintain a gap to our proposed design.
This arises from the fact that the two benchmark schemes consider
only the diagonal elements in the covariance matrix of $\boldsymbol{R}_{\Phi}$
individually designing signals at each transmitter. In contrast, the
proposed optimal solution, leveraging joint processing of sensing
data in the central controller, can effectively exploit the covariance
between different transmitters to maximize the detection probability.
\begin{figure}[H]
\centering\includegraphics[scale=0.45]{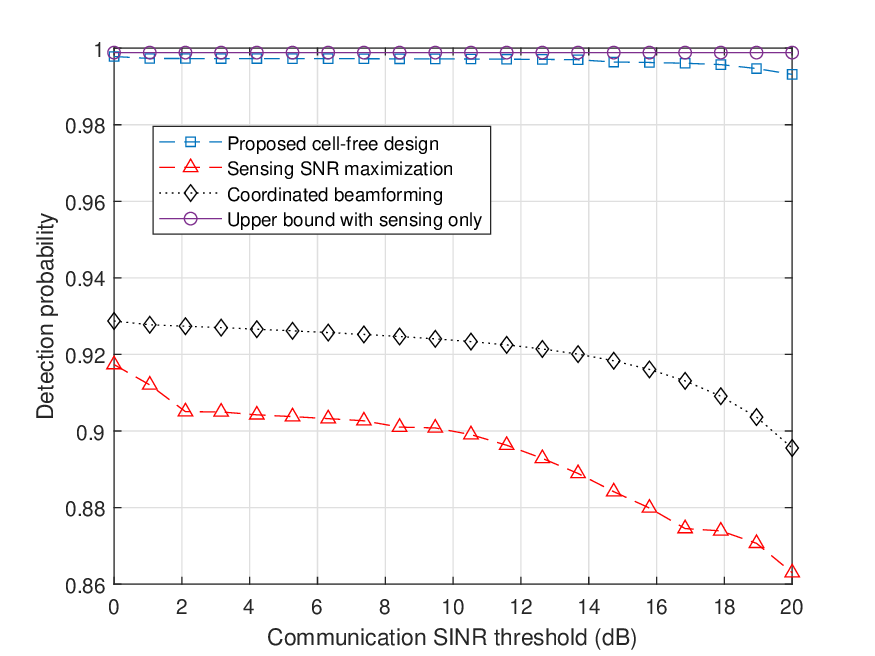}\centering\caption{\label{fig:4}The detection probability versus communication SINR
threshold $\Gamma$.}
\end{figure}

Fig. \ref{fig:4} shows the detection probability versus the communication
SINR threshold $\Gamma$ with $P=43\textrm{ dBm}$. It is observed
that the detection performance of all three methods decreases
as $\Gamma$ increases. This is due
to the fact that as the SINR requirements become more stringent, more
resource is allocated to communication to satisfy the need. Besides,
our proposed design demonstrates its excellent robustness against
the increase of communication SINR thresholds. In particular, the
performance of the SNR maximization design experiences a rapid reduction
with $\Gamma$ increasing. In contrast, the proposed design and
coordinated beamforming design are generally insensitive to the variation
of $\Gamma$. This is because that the SNR
maximization method is more sensitive to the spatial-domain energy
distribution. A higher communication threshold directly leads to less
energy at the target direction, and thus diminishes its performance. 
\begin{figure}[H]
\centering\includegraphics[scale=0.45]{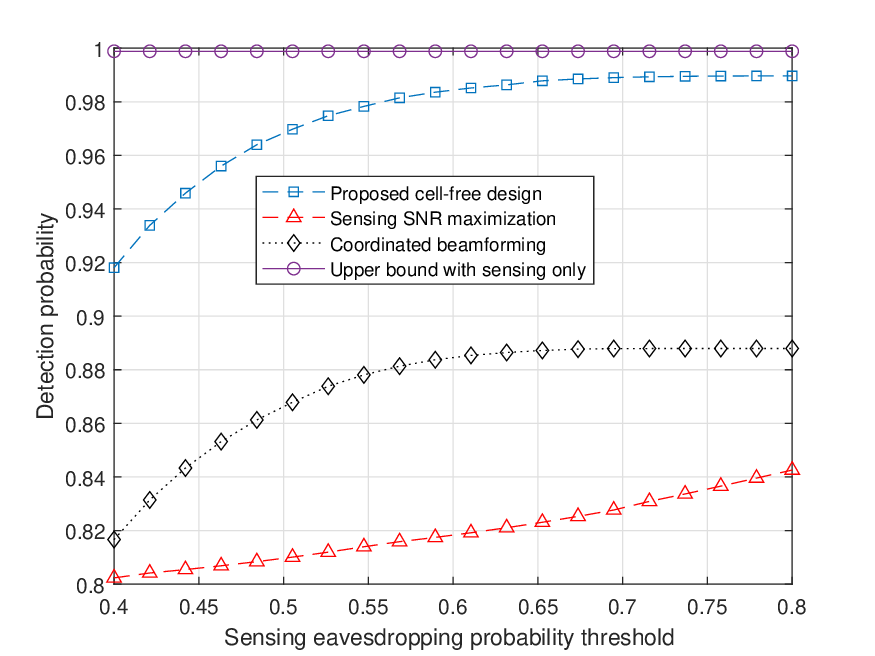}\centering\caption{\label{fig:5}The detection probability versus sensing eavesdropping
probability threshold $\Lambda$.}
\end{figure}

Fig. \ref{fig:5} shows the detection probability versus the sensing eavesdropping
probability threshold $\Lambda$ with power budget $P=40\textrm{ dBm}$.
It is evident that a more relaxed sensing eavesdropping probability
constraint or a larger value of $\Lambda$ results in a higher detection
probability. Additionally, the proposed design consistently outperforms
the coordinated beamforming design, and achieves a detection probability
that is 10\% higher than that by the sensing SNR maximization and the coordinated beamforming. This observation underscores the
superiority of the cell-free architecture.
\begin{figure}[H]
\centering\includegraphics[scale=0.45]{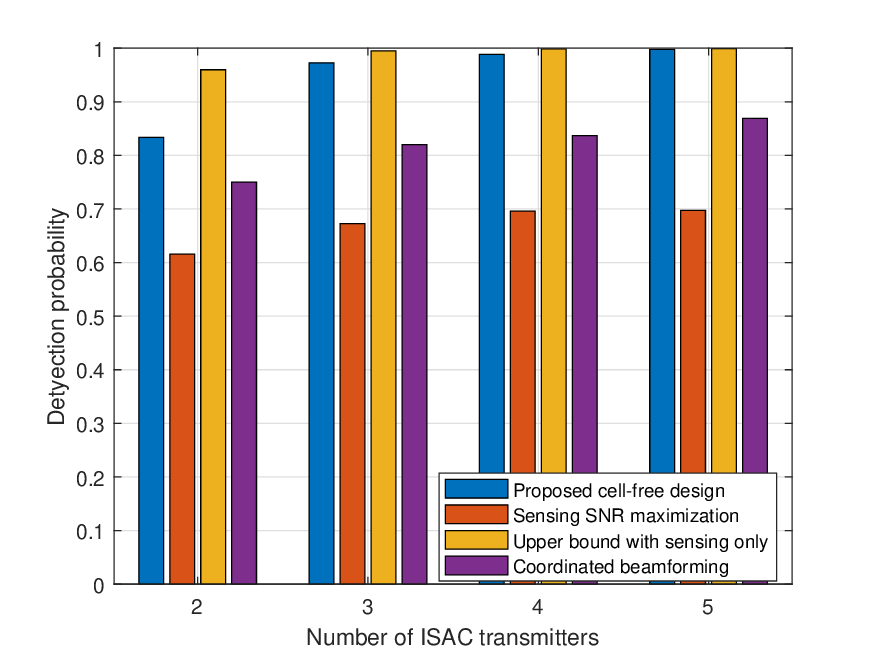}\centering\caption{\label{fig:7}The detection probability versus number of ISAC transmitters
$M_{t}$.}
\end{figure}

Fig. \ref{fig:7} shows the detection probability versus the number
of ISAC transmitters $M_{t}$, in which the ISAC transmitters are randomly located in the region, and the total transmit power is fixed with $P=43\textrm{ dBm}/M_{t}$.
In this case, it is evident that the detection
probability increases with the number of ISAC transmitters. This is attributed
to the fact that more ISAC transmitters allow for information acquisition
from different angles, providing a spatial diversity gain for the
multi-static sensing, thereby improving the multi-static sensing performance.

\section{Conclusion}

This paper investigated the joint design of transmit beamforming for
a secure cell-free ISAC system, where multiple ISAC transmitters collaboratively
serve multiple CUs while concurrently performing target detection
in a specified area. The scenario included the presence of multiple
information eavesdroppers attempting to intercept confidential CU
data and sensing eavesdroppers aiming to extract target information
from received echo signals. We formulated a transmit beamforming optimization
problem with the objective of maximizing the detection probability while
satisfying the SINR constraints for CUs, SNR constraints for information
eavesdroppers, sensing eavesdropping detection probability constraints
for sensing eavesdroppers, and transmit power constraints for each transmitter.
We obtained the global optimal solution using a SDR-based method,
with a rigorous proof of the relaxation's tightness. Furthermore,
we formulated two alternative designs based on sensing SNR maximization
and coordinated beamforming.
Numerical results validated the effectiveness of our proposed design,
as compared with the two alternative designs.

There are several interesting extensions for future research in secure
cell-free ISAC systems. First, this paper can be extended to other
scenarios with practical constraints, such as imperfect channel state information (CSI) and time synchronization among 
ISAC transmitters. Furthermore, investigating the operation of distributed
optimization in large-scale cell-free ISAC networks presents another
direction for future exploration. These extensions hold the potential
to enhance the understanding and performance of secure cell-free ISAC
systems in diverse and expansive scenarios.

\appendices{}

\section{Proof of Proposition 1}

For a given $\boldsymbol{\alpha}$, the probability density function (pdf) of $\boldsymbol{y}_{s}$
is given as 
\begin{equation}
p(\boldsymbol{r}_{\mathrm{S}};\boldsymbol{\alpha})=\frac{1}{(\pi\sigma_{s}^{2})^{NM_{r}T}}\exp\big(-\frac{1}{\sigma_{s}^{2}}(\boldsymbol{r}_{\mathrm{S}}-\boldsymbol{\Psi}\boldsymbol{\alpha})^{H}(\boldsymbol{r}_{\mathrm{S}}-\boldsymbol{\Psi}\boldsymbol{\alpha})\big).
\end{equation}
The GLRT detector is expressed as 
\begin{equation}
L_{\mathrm{G}}(\boldsymbol{r}_{\mathrm{S}})=\frac{p(\boldsymbol{r}_{\mathrm{S}};\hat{\boldsymbol{\alpha}}_{1})}{p(\boldsymbol{r}_{\mathrm{S}};\hat{\boldsymbol{\alpha}}_{0})}\gtreqless\Gamma,
\end{equation}
where $\Gamma$ is a given decision threshold w.r.t. the false alarm
probability, and $\hat{\boldsymbol{\alpha}}_{1}$ and $\hat{\boldsymbol{\alpha}}_{0}$
are the maximum likelihood estimation (MLE) for $\boldsymbol{\alpha}$
under conditions $\mathcal{H}_{1}$ and $\mathcal{H}_{0}$, respectively.
Specifically, $\hat{\boldsymbol{\alpha}}_{1}$ is an unconstrained
MLE calculated as 
\begin{equation}
\hat{\boldsymbol{\alpha}}_{1}=(\boldsymbol{\Psi}^{H}\boldsymbol{\Psi})^{-1}\boldsymbol{\Psi}^{H}\boldsymbol{r}_{\mathrm{S}}.
\end{equation}
On the other hand, $\hat{\boldsymbol{\alpha}}_{0}$ is a constrained
MLE subject to the constraint $\boldsymbol{\mu}^{T}\hat{\boldsymbol{\alpha}}_{0}=0$,
which can be derived via the Lagrange multiplier method \cite{kay1993fundamentals}.
The problem of calculating $\hat{\boldsymbol{\alpha}}_{0}$ is given
as 
\begin{eqnarray*}
\underset{\hat{\boldsymbol{\alpha}}_{0}}{\min} & \|\boldsymbol{r}_{\mathrm{S}}-\boldsymbol{\Psi}\hat{\boldsymbol{\alpha}}_{0}\|^{2} & \textrm{s.t. }\boldsymbol{\mu}^{T}\hat{\boldsymbol{\alpha}}_{0}=0.
\end{eqnarray*}
The Lagrangian is given as 
\begin{equation}
L(\hat{\boldsymbol{\alpha}}_{0})=\|\boldsymbol{r}_{\mathrm{S}}-\boldsymbol{\Psi}\hat{\boldsymbol{\alpha}}_{0}\|^{2}+\rho\boldsymbol{\mu}^{T}\hat{\boldsymbol{\alpha}}_{0},
\end{equation}
where $\rho$ represents the Lagrange multiplier for constraint $\boldsymbol{\mu}^{T}\hat{\boldsymbol{\alpha}}_{0}=0$.
By checking the derivation of $L(\hat{\boldsymbol{\alpha}}_{0})$,
we have 
\begin{eqnarray}
\frac{\partial L(\hat{\boldsymbol{\alpha}}_{0})}{\partial\hat{\boldsymbol{\alpha}}_{0}} & = & -2\boldsymbol{\Psi}^{H}\boldsymbol{r}_{\mathrm{S}}+2\boldsymbol{\Psi}^{H}\boldsymbol{\Psi}\hat{\boldsymbol{\alpha}}_{0}+\rho\boldsymbol{\mu}.\label{eq:derivative}
\end{eqnarray}
By setting the derivative in \eqref{eq:derivative} equal to $\boldsymbol{0}$,
we have 
\begin{equation}
\hat{\boldsymbol{\alpha}}_{0}=\underset{\hat{\boldsymbol{\alpha}}_{1}}{\underbrace{(\boldsymbol{\Psi}^{H}\boldsymbol{\Psi})^{-1}\boldsymbol{\Psi}^{H}\boldsymbol{r}_{\mathrm{S}}}}-(\boldsymbol{\Psi}^{H}\boldsymbol{\Psi})^{-1}\frac{\rho\boldsymbol{\mu}}{2}.
\end{equation}
Then, we can easily calculate $\rho\geq0$ as 
\begin{equation}
\begin{aligned} & \boldsymbol{\mu}^{T}\hat{\boldsymbol{\alpha}}_{0}=0\\
\Rightarrow & \boldsymbol{\mu}^{T}[(\boldsymbol{\Psi}^{H}\boldsymbol{\Psi})^{-1}\boldsymbol{\Psi}^{H}\boldsymbol{r}_{\mathrm{S}}-(\boldsymbol{\Psi}^{H}\boldsymbol{\Psi})^{-1}\frac{\rho\boldsymbol{\mu}}{2}]=0\\
\Rightarrow & \frac{\rho}{2}=\frac{\boldsymbol{\mu}^{T}(\boldsymbol{\Psi}^{H}\boldsymbol{\Psi})^{-1}\boldsymbol{\Psi}^{H}\boldsymbol{r}_{\mathrm{S}}}{\boldsymbol{\mu}^{T}(\boldsymbol{\Psi}^{H}\boldsymbol{\Psi})^{-1}\boldsymbol{\mu}}=\frac{\boldsymbol{\mu}^{T}\hat{\boldsymbol{\alpha}}_{1}}{\boldsymbol{\mu}^{T}(\boldsymbol{\Psi}^{H}\boldsymbol{\Psi})^{-1}\boldsymbol{\mu}}.
\end{aligned}
\end{equation}
 Consequently, $\hat{\boldsymbol{\alpha}}_{0}$ is given as 
\begin{equation}
\hat{\boldsymbol{\alpha}}_{0}=\hat{\boldsymbol{\alpha}}_{1}-\frac{(\boldsymbol{\Psi}^{H}\boldsymbol{\Psi})^{-1}\boldsymbol{\mu}^{T}\hat{\boldsymbol{\alpha}}_{1}\boldsymbol{\mu}}{\boldsymbol{\mu}^{T}(\boldsymbol{\Psi}^{H}\boldsymbol{\Psi})^{-1}\boldsymbol{\mu}}.
\end{equation}
As a result, the expression for $\ln(L_{\mathrm{G}}(\boldsymbol{y}_{s}))$
is
\begin{eqnarray}
\ln\big(L_{\mathrm{G}}(\boldsymbol{r}_{\mathrm{S}})\big) & = & -\frac{1}{\sigma_{s}^{2}}[(\boldsymbol{r}_{\mathrm{S}}-\boldsymbol{\Psi}\hat{\boldsymbol{\alpha}}_{1})^{H}(\boldsymbol{r}_{\mathrm{S}}-\boldsymbol{\Psi}\hat{\boldsymbol{\alpha}}_{1})\nonumber \\
 &  & -(\boldsymbol{r}_{\mathrm{S}}-\boldsymbol{\Psi}\hat{\boldsymbol{\alpha}}_{0})^{H}(\boldsymbol{r}_{\mathrm{S}}-\boldsymbol{\Psi}\hat{\boldsymbol{\alpha}}_{0})].
\end{eqnarray}
After some mathematical derivation, we have 
\begin{equation}
\ln\big(L_{\mathrm{G}}(\boldsymbol{r}_{\mathrm{S}})\big)=\frac{|\boldsymbol{\mu}^{T}\hat{\boldsymbol{\alpha}}_{1}|^{2}}{\sigma_{s}^{2}\boldsymbol{\mu}^{T}(\boldsymbol{\Psi}^{H}\boldsymbol{\Psi})^{-1}\boldsymbol{\mu}}.
\end{equation}
Then we calculate the distribution of $\ln\big(L_{\mathrm{G}}(\boldsymbol{r}_{\mathrm{S}})\big)$
under two different conditions. Notice that $\hat{\boldsymbol{\alpha}}_{1}\sim\mathcal{CN}\big(\boldsymbol{\alpha},\sigma_{s}^{2}(\boldsymbol{\Psi}^{H}\boldsymbol{\Psi})^{-1}\big)$, and accordingly
we have 
\begin{equation}
\boldsymbol{\mu}^{T}\hat{\boldsymbol{\alpha}}_{1}\sim\begin{cases}
\mathcal{CN}\big(\boldsymbol{0},\sigma_{s}^{2}\boldsymbol{\mu}^{T}(\boldsymbol{\Psi}^{H}\boldsymbol{\Psi})^{-1}\boldsymbol{\mu}\big), & \mathcal{H}_{0},\\
\mathcal{CN}\big(\boldsymbol{\mu}^{T}\boldsymbol{\alpha},\sigma_{s}^{2}\boldsymbol{\mu}^{T}(\boldsymbol{\Psi}^{H}\boldsymbol{\Psi})^{-1}\boldsymbol{\mu}\big), & \mathcal{H}_{1}.
\end{cases}
\end{equation}
As a result, the distribution of $\ln\big(L_{\mathrm{G}}(\boldsymbol{r}_{\mathrm{S}})\big)$
is given as 
\begin{equation}
\ln(L_{\mathrm{G}}(\boldsymbol{y}_{s}))\sim\begin{cases}
\chi^{2}, & \mathcal{H}_{0},\\
\tilde{\chi}^{2}(\lambda), & \mathcal{H}_{1},
\end{cases}
\end{equation}
where $\lambda=\frac{|\boldsymbol{\mu}^{T}\boldsymbol{\alpha}|^{2}}{\sigma_{s}^{2}\boldsymbol{\mu}^{T}(\boldsymbol{\Psi}^{H}\boldsymbol{\Psi})^{-1}\boldsymbol{\mu}}$.
This completes the proof.

\section{Proof of Proposition 2}

Based on the Neyman-Pearson criterion, the optimal detection rule
for the sensing eavesdropper is given as 
\begin{equation}
L_{\mathrm{G}}(\tilde{r}_{q}(t))=\frac{p_{0}(\tilde{r}_{q}(t))}{p_{1}(\tilde{r}_{q}(t))}\gtreqless1,
\end{equation}
which can be expressed as 
\begin{equation}
|\tilde{r}_{q}(t)|^{2}\gtreqless\frac{\zeta_{q}(\{\boldsymbol{w}_{k}\},\boldsymbol{S})\beta_{q}(\{\boldsymbol{w}_{k}\},\boldsymbol{S})}{\beta_{q}(\{\boldsymbol{w}_{k}\},\boldsymbol{S})-\zeta_{q}(\{\boldsymbol{w}_{k}\},\boldsymbol{S})}\ln(\frac{\beta_{q}(\{\boldsymbol{w}_{k}\},\boldsymbol{S})}{\zeta_{q}(\{\boldsymbol{w}_{k}\},\boldsymbol{S})}).\label{eq:70}
\end{equation}
According to \eqref{eq:pdf2}, the cumulative density functions (CDFs)
under $\mathcal{\tilde{H}}_{0}$ and $\mathcal{\tilde{H}}_{1}$ are respectively
given as 
\begin{eqnarray}
\mathrm{Pr}(|\tilde{r}_{q}(t)|^{2}\big|\mathcal{\tilde{H}}_{0}) & = & 1-\exp(-\frac{|\tilde{r}_{q}(t)|^{2}}{\zeta_{q}(\{\boldsymbol{w}_{k}\},\boldsymbol{S})}),\\
\mathrm{Pr}(|\tilde{r}_{q}(t)|^{2}\big|\mathcal{\tilde{H}}_{1}) & = & 1-\exp(-\frac{|\tilde{r}_{q}(t)|^{2}}{\beta_{q}(\{\boldsymbol{w}_{k}\},\boldsymbol{S})}).\label{eq:72}
\end{eqnarray}
By substituting \eqref{eq:70} into \eqref{eq:72}, we obtain 
\begin{equation}
\ensuremath{\tilde{p}_{q}(\{\boldsymbol{w}_{k}\},\boldsymbol{S})}=(\frac{\beta_{q}(\{\boldsymbol{w}_{k}\},\boldsymbol{S})}{\zeta_{q}(\{\boldsymbol{w}_{k}\},\boldsymbol{S})})^{-\frac{\zeta_{q}(\{\boldsymbol{w}_{k}\},\boldsymbol{S})}{\beta_{q}(\{\boldsymbol{w}_{k}\},\boldsymbol{S})-\zeta_{q}(\{\boldsymbol{w}_{k}\},\boldsymbol{S})}}.
\end{equation}
This completes the proof.

\section{Proof of Proposition 3}

First, we consider the optimal objective function value. The optimal
covariance matrix is given as\vspace{-0.1cm}
\begin{equation}
\boldsymbol{R_{x}}^{*}=\sum_{k=1}^{K}\boldsymbol{W}_{k}^{\mathrm{opt}}+\boldsymbol{S}^{\mathrm{opt}}=\sum_{k=1}^{K}\boldsymbol{W}_{k}^{*}+\boldsymbol{S}^{*}.
\end{equation}
This implies that the objective function values achieved by $\{\boldsymbol{W}_{k}^{\mathrm{opt}}\}$
and $\boldsymbol{S}^{\mathrm{opt}}$ are equivalent to that of $\{\boldsymbol{W}_{k}^{*}\}$
and $\boldsymbol{S}^{*}$. Furthermore, the constraints specified
in \eqref{eq:sensing eavesdropper} and \eqref{eq:rank} are satisfied.
Subsequently, by observing that the SINR constraints in \eqref{eq:SINR1}
are satisfied due to the equality $\boldsymbol{h}_{k}^{H}\boldsymbol{W}_{k}^{*}\boldsymbol{h}_{k}=\boldsymbol{h}_{k}^{H}\boldsymbol{W}_{k}^{\mathrm{opt}}\boldsymbol{h}_{k}$,
$\forall k\in\mathcal{K}$. Then, by letting $\boldsymbol{v}\in\mathbb{C}^{NM_{t}\times1}$
denote an arbitrary vector, we have\vspace{-0.2cm} 
\begin{eqnarray}
\boldsymbol{v}^{H}(\boldsymbol{W}_{k}^{*}-\boldsymbol{W}_{k}^{\mathrm{opt}})\boldsymbol{v} & = & \boldsymbol{v}^{H}\boldsymbol{W}_{k}^{*}\boldsymbol{v}-\frac{|\boldsymbol{v}^{H}\boldsymbol{W}_{k}^{*}\boldsymbol{h}_{k}|^{2}}{\boldsymbol{h}_{k}^{H}\boldsymbol{W}_{k}^{*}\boldsymbol{h}_{k}}.
\end{eqnarray}
Based on the Cauchy-Schwarz inequality \cite{wang2023globally}, we
have $(\boldsymbol{v}^{H}\boldsymbol{W}_{k}^{*}\boldsymbol{v})(\boldsymbol{h}_{k}^{H}\boldsymbol{W}_{k}^{*}\boldsymbol{h}_{k})\geq|\boldsymbol{v}^{H}\boldsymbol{W}_{k}^{*}\boldsymbol{h}_{k}|^{2}$.
As a result, $\boldsymbol{W}_{k}^{*}-\boldsymbol{W}_{k}^{\mathrm{opt}}\succeq\boldsymbol{0}$,
$\forall k\in\mathcal{K}$. Consequently, we have 
\begin{equation}
\boldsymbol{g}_{l}^{H}\boldsymbol{W}_{k}^{\mathrm{opt}}\boldsymbol{g}_{l}\leq\boldsymbol{g}_{l}^{H}\boldsymbol{W}_{k}^{*}\boldsymbol{g}_{l}\leq\Lambda\sigma^{2},\forall l\in\mathcal{L}.
\end{equation}
Hence, the constraints in \eqref{eq:SNR1} are also satisfied. This
completes the proof.

{\footnotesize{}\bibliographystyle{IEEEtran}
\bibliography{IEEEabrv,IEEEexample,my_ref}
}{\footnotesize\par}
\end{document}